\RequirePackage{fix-cm}

\documentclass{svjour3-mod} 
\smartqed  
\usepackage{graphicx}
\usepackage{latexsym}
\usepackage{natbib}
\usepackage{url,xspace}
\usepackage{proof}
\usepackage{cite}
\usepackage{listings}

\newbox\tempa
\newbox\tempb
\newdimen\tempc
\newbox\tempd

\def\mud#1{\hfil $\displaystyle{#1}$\hfil}
\def\rig#1{\hfil $\displaystyle{#1}$}

\def\inruleanhelp#1#2#3{\setbox\tempa=\hbox{$\displaystyle{\mathstrut #2}$}%
                        \setbox\tempd=\hbox{$\; #3$}%
                        \setbox\tempb=\vbox{\halign{##\cr
        \mud{#1}\cr
        \noalign{\vskip\the\lineskip}%
        \noalign{\hrule height 0pt}%
        \rig{\vbox to 0pt{\vss\hbox to 0pt{\copy\tempd \hss}\vss}}\cr
        \noalign{\hrule}%
        \noalign{\vskip\the\lineskip}%
        \mud{\copy\tempa}\cr}}%
                      \tempc=\wd\tempb
                      \advance\tempc by \wd\tempa
                      \divide\tempc by 2 }

\def\inrulean#1#2#3{{\inruleanhelp{#1}{#2}{#3}%
                     \hbox to \wd\tempa{\hss \box\tempb \hss}}}

\def\lowerhalf#1{\hbox{\raise -0.8\baselineskip\hbox{#1}}}

\def\ianc#1#2#3{{\lineskip 4pt\lowerhalf{\inruleanhelp{#1}{#2}{#3}%
                   \box\tempb\hskip\wd\tempd}}}
\def\ibnc#1#2#3#4{{\lineskip 4pt\ianc{#1\quad\qquad #2}{#3}{#4}}}

\def\bnfas{\mathrel{::=}}

\def\lam{\lambda}

\def\oftp{\mathord{:}}
\def\hastype{\mathrel{:}}
\def\ldot{\mathord{.}\;}
\long\def\ignore#1{}
\def\vd{\null\mathrel{\vdash}}

\newcommand{\FV}{\textsf{EV}}
\renewcommand{\cite}{\citep}
\newcommand{\cD}{{\cal D}}
\newcommand{\crel}[2]{#1\sim#2}

\newcommand{\termeq}[2]{\mathsf{aeq}\;#1\;#2}
\newcommand{\tpaeq}[2]{\mathsf{atp}\;#1\;#2}
 \newcommand{\istm}[1]{\mathsf{is\_tm}\;#1}
 \newcommand{\istp}[1]{\mathsf{is\_tp}\;#1}

\newcommand{\termequal}[2]{\mathsf{deq}\;#1\;#2}

\newcommand{\app}[2]{\mathsf{app}\; #1\; #2}
\newcommand{\capp}{\mathsf{app}}
\newcommand{\tapp}[2]{\mathsf{tapp}\; #1\; #2}
\renewcommand{\lam}[2]{\mathsf{lam}\,#1.\,#2}
\newcommand{\clam}{\mathsf{lam}}
\newcommand{\lamt}[3]{\mathsf{lam}\,#1^{}.\,#3}
\newcommand{\tlam}[2]{\mathsf{tlam}\,#1.\,#2}
\newcommand{\all}[2]{\mathsf{all}\,#1.\,#2}
\newcommand{\call}{\mathsf{all}}
\newcommand{\grnd}{\ensuremath{\mathsf{i}}}

\newcommand{\alt}{\, \ensuremath{\mathsf{+}} \, }

\spnewtheorem{define}[theorem]{Definition}{\bfseries}{\itshape}
\spnewtheorem{prop}[theorem]{Property}{\bfseries}{\itshape}
\spnewtheorem{lem}[theorem]{Lemma}{\bfseries}{\itshape}

 \newcommand{\oft}[2]{\mathtt{oft}\;#1\;#2}

\newcommand{\lcase}[1]{\hfill\mbox{#1}\\}

\newcommand{\caseof}[2]{\textsf{case}\;#1\;\textsf{of}\;#2}

\newcommand{\step}{\leadsto}
\newcommand{\fsp}[2]{\mathsf{arr} #1 \, #2}
\newcommand{\carr}{\mathsf{arr}}
\newcommand{\tctx}{\Phi_t}
\newcommand{\tsch}{S_t}
\newcommand{\rctx}{\Phi_r}
\newcommand{\rsch}{S_r}
\newcommand{\ictx}{\Phi_x}
\newcommand{\isch}{S_x}
\newcommand{\actx}{\Phi_{xa}}
\newcommand{\asch}{S_{xa}}
\newcommand{\dctx}{\Phi_{xd}} 
\newcommand{\dsch}{S_{xd}} 
 
\newcommand{\pctx}{\Phi_{a\!e\!q}}
\newcommand{\psch}{S_{a\!e\!q}}
\newcommand{\alphctx}{\Phi_{\alpha}}
\newcommand{\alphsch}{S_{\alpha}}
\newcommand{\axctx}{\Phi_{\alpha x}}
\newcommand{\axsch}{S_{\alpha x}}
\newcommand{\atctx}{\Phi_{\alpha t}}

\newcommand{\atpctx}{\Phi_{{\alpha t}}}
\newcommand{\atpsch}{S_{\alpha t}}
\newcommand{\atpeqctx}{\Phi_{{a\! t\! p}}}
\newcommand{\atpeqsch}{S_{{a\! t\! p}}}
\newcommand{\vsk}{\\[2ex]}

\def\gvd{\tctx \vd}
\def\pvd{\rctx \vd}

\newcommand{\schema}[2]{#1\textsf{ has\_schema }#2}
\newcommand{\regdecl}[2]{#1\in #2}
\newcommand{\memb}[2]{#1 \in #2}

\newcommand{\copda}[2]{\textit{c-#1}}
\newcommand{\dop}[1]{\textit{d-#1}}

\newcommand{\aboved}[2]{
  \begin{array}{c}
    #1\\#2
  \end{array}
}

 \newcommand{\wf}[1]{\mathit{is\_wf}\;#1}

\journalname{Journal of Automated Reasoning}

\begin{document}

\title{The Next 700 Challenge Problems for Reasoning with Higher-Order Abstract
  Syntax Representations}
\subtitle{Part 1---A Common Infrastructure for Benchmarks}

\titlerunning{The Next 700 Challenge Problems: A Common Infrastructure}

\author{Amy P. Felty \and Alberto Momigliano \and Brigitte Pientka}

\institute{A. P. Felty \at
           School of Electrical Engineering and Computer Science,
           University of Ottawa, Ottawa, Canada,
           \email{afelty@eecs.uottawa.ca}
           \and
           A. Momigliano \at
           Dipartimento di Informatica, Universit\`{a}
           degli Studi di Milano, Italy,
           \email{momigliano@di.unimi.it}
          \and
           B. Pientka \at
           School of Computer Science, McGill University, Montreal,
           Canada,
           \email{bpientka@cs.mcgill.ca}
}

\date{\today}

\maketitle

\begin{abstract}
  A variety of logical frameworks support the use of higher-order
  abstract syntax (HOAS) in representing formal systems.  Although
  these systems seem superficially the same, they differ in a variety
  of ways; for example, how they handle a \emph{context} of
  assumptions and which theorems about a given formal system can be
  concisely expressed and proved. Our contributions in this paper are
  three-fold: 1) we develop a common infrastructure for representing
  benchmarks for systems supporting reasoning with binders, 2) we
  present several concrete benchmarks, which highlight a variety of
  different aspects of reasoning within a context of assumptions, and
  3) we design an open repository \emph{ORBI} (\underline{O}pen
  challenge problem \underline{R}epository for systems supporting
  reasoning with \underline{BI}nders).  Our work sets the stage for
  providing a basis for qualitative comparison of different systems.
  This allows us to review and survey the state of the art, which we
  do in great detail for four systems in Part 2 of this
  paper~\cite{companion}.  It also allows us to outline future
  fundamental research questions regarding the design and
  implementation of meta-reasoning systems.

\keywords{Logical Frameworks \and Higher-Order Abstract Syntax \and
    Context Reasoning \and Benchmarks}
\end{abstract}

\section{Introduction}
In recent years the \textsc{PoplMark} challenge \cite{Aydemir05TPHOLS} has
stimulated considerable interest in mechanizing the meta-theory of
programming languages and it has played a substantial role in the wide-spread
use of proof assistants to prove properties, for example, of parts of a compiler or  of
a language design.  
The \textsc{PoplMark} challenge concentrated on summarizing the state of the
art, identifying best practices for (programming language)
researchers embarking on formalizing language definitions, and
identifying a list of engineering improvements to make the use of
proof assistants (more) common place. While these are important questions
whose answers will foster the adoption of proof assistants  by
non-experts, it neglects some of the deeper fundamental questions:
 What should existing or future
meta-languages and meta-reasoning environments look like and what
requirements should they satisfy?
What support should an ideal
meta-language and proof environment give to facilitate mechanizing
meta-reasoning? How can its design 
reflect and support these ideals? 

We believe ``good'' meta-languages should free the user from dealing
with tedious bureaucratic details, so s/he is able to concentrate on
the essence of a proof or algorithm. Ultimately, this means that users
will mechanize proofs more quickly.  In addition, since effort is not
wasted on cumbersome 
details, proofs are more likely to capture only the essential steps of
the reasoning process, and as a result, may be easier to trust.
For instance, weakening is a typical a low-level lemma that is
used pervasively (and silently) in a proof.  Freeing the user of such details
ultimately may also mean that the automation of such proofs is more
feasible.

One fundamental question when mechanizing formal systems and their
meta-theory is how to represent variables and variable binding
structures. There is a wide range of answers to this question from
using de Bruijn indices to locally nameless representations,
 and  nominal
encodings, etc. For a partial view of the field see the papers
collected in the \emph{Journal of Automated Reasoning}'s special issue
dedicated to \textsc{PoplMark} \cite{POPL-JAR} and the one on ``Abstraction,
Substitution and Naming'' \cite{NOM-JAR}.

Encoding object languages and logics (OLs) via higher-order abstract syntax
(HOAS), sometimes referred to as ``lambda-tree syntax''~\cite{MillerP99},
where we utilize meta-level binders to model object-level binders is
in our opinion the most advanced technology.   HOAS 
avoids implementing common although notoriously tricky routines
dealing with variables, such as capture-avoiding substitution,
renaming, and fresh name generation. Compared to other techniques,
HOAS leads to very concise and elegant encodings and provides
significant support for such an endeavor.  Concentrating on encoding
binders, however, neglects another important and fundamental aspect:
the support for hypothetical and parametric reasoning, in other words
reasoning within a context of assumptions. Considering a derivation
within a context is common place in programming language theory and
leads to several natural questions:
How do we model the context of assumptions?  How do we know that a
derivation is sensible within the scope of a context?  Can we model
the relationships between different contexts?  How do we deal with
structural properties of contexts such as weakening, strengthening,
and exchange?
How do we know assumptions in a context occur uniquely?
How do we take advantage of the HOAS approach to substitution?

Even in systems supporting HOAS there is not a uniform answer to
these questions.  On one side of the spectrum  we have systems that
implement various dependently-typed calculi. 
Such systems include the logical framework Twelf \cite{TwelfSP}, the
dependently-typed functional language Beluga
\cite{Pientka:POPL08,Pientka:IJCAR10}, and
Delphin~\cite{Schuermann:ESOP08}.  All these systems also provide, in
various degrees, built-in support for reasoning modulo structural
properties of a context of assumptions.

On the other side  there are systems based on a
proof-theoretic foundation, which follow a two-level approach: they
implement a specification logic (SL) inside a higher-order logic or
type theory. Hypothetical judgments of object languages are modeled using
implication in the SL and parametric judgments are handled via
(generic) universal quantification. 
Contexts are commonly represented explicitly as lists or sets in the
SL, and structural properties are established separately as
lemmas. For example substituting for an assumption is justified by
appealing to the cut-admissibility lemma of the SL.  These lemmas are
not directly and intrinsically supported through the SL, but may be
integrated into a system's automated proving procedures, usually via
tactics.  Systems following this philosophy are for instance the
two-level Hybrid system~\cite{MMF07,FeltyMomigliano:JAR10} as
implemented on top of Coq
 and
Isabelle/HOL,
and the Abella system \cite{Gacek:IJCAR08}.

This paper, together with Part 2~\cite{companion}, is a major
extension of an earlier conference paper \cite{felty/pientka:ITP2011}.
The contributions of the present paper are three-fold.  {First}, we
develop a common framework and infrastructure for representing
\emph{benchmarks} for systems supporting reasoning with binders; in
particular, we develop notation to 
view contexts as ``structured sequences'' and classify contexts using
schemas. Moreover, we abstractly characterize in a uniform way basic
structural properties that many object languages satisfy, such as
weakening, strengthening, and exchange. This lays the foundation for
describing benchmarks and comparing different approaches to
mechanizing OLs. {Second}, we propose several challenge problems that
are \emph{crafted} to highlight the differences between the designs of
various meta-languages with respect to reasoning with and within a
context of assumptions, in view of their mechanization in a given
proof assistant. In Part 2 of this paper~\cite{companion}, we carry
out such a comparison on four systems: Twelf, Beluga, Hybrid, and
Abella.  {Third}, we discuss the design of \emph{ORBI}
(\underline{O}pen challenge problem \underline{R}epository for systems
supporting reasoning with \underline{BI}nders), an open repository for
sharing benchmark problems based on the infrastructure that we have
developed.  Although \emph{ORBI}'s syntax is inspired by systems such
as Twelf and Beluga, we do not commit to using a particular system, as
we wish to retain the needed flexibility to be able to easily support
translations to both type-theoretic and proof-theoretic
approaches.\footnote{A first step in this direction is the translator
  for Hybrid, whose first version is presented
  in~\citet{HabliFelty:PXTP13}.} The common notation allows us to
express the syntax of object languages 
 that we wish to reason about, as well as the
context schemas, the judgments and inference rules, and the statements
of the benchmark theorems.  We hope that ORBI will foster sharing of
examples in the community and provide a common set of
examples. We also see our benchmark repository as a place to collect and propose
``open'' challenge problems to push the development of meta-reasoning
systems.  

The challenge problems also play a role in allowing us, as designers
and developers of logical frameworks, to highlight and explain how the
design decisions for each individual system lead to differences in
using them in practice. This means reviewing the state of the art, as
well as outlining future fundamental research questions regarding the
design and implementation of meta-reasoning systems, as we discuss
further in the companion paper~\cite{companion}.  Additionally, our
benchmarks aim to provide a better understanding of what practitioners
should be looking for, as well as help them foresee what kind of
problems can be solved elegantly and easily in a given system, and
more importantly, why this is the case. Therefore the challenge
problems provide guidance for users and developers in better
comprehending differences and limitations. Finally, they serve as an
excellent regression suite.

%

This paper does not, of course, present 700 challenge
problems.  We start with a few and hope that others will contribute to
the benchmark repository, implement these challenge problems, and
further our understanding of the trade-offs involved in choosing one
system over another for this kind of reasoning.

The paper is structured as follows: In Sect.~\ref{sec:theory} we
motivate our definition of contexts as ``structured sequences,'' which
refines the standard view of contexts, and we describe generically and
abstractly some context properties.
Using this terminology we then present the benchmarks and their proofs
in Sect.~\ref{sec:bench}.  In Sect.~\ref{sec:mech}, 
we introduce ORBI and discuss how it provides HOAS encodings of the
benchmarks in a uniform manner. 
We discuss related work in Sect.~\ref{sec:related},
before concluding in Sect.~\ref{sec:concl}.
Appendix~\ref{sec:overview} provides a quick reference guide
to the benchmarks and Appendix~\ref{sec:aede-orbi} gives a complete
example of an ORBI file for a selection of the benchmark problems.
Full details about the challenge problems and their mechanization can
be found at \url{https://github.com/pientka/ORBI}. The latter, as well
as the present paper, can be better appreciated by reading the
companion paper \cite{companion}.


\section{Contexts of Assumptions: Preliminaries and Terminology}
\label{sec:theory}

Reasoning with and within a context of assumptions is common when we prove
meta-theoretic properties about object languages such as type systems or
logics. Hence, how to represent contexts and enforce properties such as
well-formedness, weakening, strengthening, exchange, uniqueness of assumptions, 
and substitution is a central issue once we mechanize such reasoning.

As mentioned, proof environments supporting higher-order abstract
syntax differ in how they represent and model contexts and our
comparison~\cite{companion} to a large extent focuses on this
issue. Here we lay down a common framework and notation for describing
the syntax of object languages, inference rules and contexts by using
different representative examples. In particular, we refine the
standard view of contexts as sequences of assumptions and abstractly
describe structural properties such as weakening and exchange
satisfied by many object languages. Our description follows
mathematical practice, in contrast to giving a fully formal account
based on, for example, type theory. In fact, all the notions that we
touch upon in this section, such as substitution, $\alpha$-renaming,
bindings, context schemas to name a few, can and have been generally
treated in Beluga \citep[see e.g.,][]{Pientka:POPL08}.
However, we deliberately choose to base our description on
mathematical practice to make our benchmarks more accessible to a
wider audience and so as not to force upon us one particular
foundation. This infrastructure may be seen as a first step towards
developing a formal translation between different foundations, i.e., a
translation between Beluga's type-theoretic foundation and the
proof-theory underlying systems such as Hybrid or Abella.

\subsection{Defining Well-formed Objects} 
\label{subsec:theory}
The first question that we face when defining an OL is how
to describe well-formed objects. 
Consider the polymorphic lambda-calculus. Commonly the grammar of this language
is defined using Backus-Naur form (BNF) as follows.
\[
\begin{array}{llcl}
\mbox{Types}& A,B & \bnfas & \alpha \mid \fsp A B \mid \all \alpha A \\
\mbox{Terms} & M & \bnfas &  x \mid \lamt{x}{A}{M} \mid \app{M_1}{M_2} \mid \tlam{\alpha}{M} \mid \tapp{M}{A}
\end{array}
\]
The grammar, however, does not capture properties of interest such as when a given
term or type is \emph{closed}. Alternatively, we can describe well-formed
types and terms as \emph{judgments} using axioms and inference rules
following \citet{MartinLof85}, as popularized in programming language
theory by Pfenning's \emph{Computation and Deduction} notes~\cite{Pfenning01book}.

We start with an \emph{implicit-context} version of the rules for
well-formed types and terms that plays the part of the above BNF
grammar, but is also significantly more expressive. To describe
whether a type $A$ or term $M$ is well-formed we use two judgments:
\fbox{$\istp A$} and \fbox{$\istm M$}, whose formation rules 
are depicted in Fig.~\ref{fig:wftp}. The rule for function types
($tp_{ar}$) is unsurprising. The rule $tp_{al}$ states that a type
$\all{\alpha}A$ is well-formed if $A$ is well-formed under the
assumption that the variable $\alpha$ is also. We say that this
rule 
is \emph{parametric} in the name of the bound variable $\alpha$---thus
implicitly enforcing the usual eigenvariable condition, since bound
variables can be $\alpha$-renamed at will---and \emph{hypothetical} in
the name of the axiom ($tp_v$) stating the well-formedness of this
type variable.  In this two-dimensional representation, derived from
Gentzen's presentation of natural deduction, we do not have an
explicit rule for variables: instead, for each type variable
introduced by $tp_{al}$ we also introduce the well-formedness
assumption about that variable, and we explicitly include names for
the bound variable and axiom as parameters to the rule name.

%
\begin{figure}[th]
  \centering 
 \[
\label{sec:cpolylam}
 \begin{array}{ll}
 \multicolumn{2}{l}{\mbox{\fbox{$\istp A$}
     -- Type $A$ is well-formed}}\\[0.75em]
 \infer[tp_{al}^{\alpha,tp_v}]{\istp{(\all{\alpha}{A})}}{
  \begin{array}{l}
    \infer[tp_v]{\istp{\alpha}}{}\\
  \quad\vdots\\
 \istp{A}
  \end{array}
} 
\qquad\qquad & \qquad\qquad
 \infer[tp_{ar}]{\istp{(\fsp A B)}}
 {\istp{A}\qquad  \istp{B}
 }\\[1em]
\multicolumn{2}{l}{\mbox{\fbox{$\istm M$} -- Term $M$ is well-formed}}\\[0.75em]
 \infer[tm_l^{x,tm_v}]{\istm{(\lamt{x}{A}{M})}}{
& 
  \begin{array}{l}
    \infer[tm_v]{\istm{x}}{}\\
 \quad\vdots\\
 \istm{M}
  \end{array}
}
\qquad \qquad& \qquad\qquad
 \infer[tm_{tl}^{\alpha,tp_v}]{\istm{(\tlam{\alpha}{M})}}{
  \begin{array}{l}
    \infer[tp_v]{\istp{\alpha}}{}\\
\quad\vdots\\
\istm{M}
  \end{array}
} 
\\[1em]
 \infer[tm_a]{\istm{(\app{M_1}{M_2})}}
 {\istm{M_1}\qquad  \istm{M_2}
 }
\qquad \qquad & \qquad\qquad
 \infer[tm_{ta}]{\istm{(\tapp{M}{A})}}
 {\istm{M}\qquad  \istp{A}
 }
\end{array}
\]
  \caption{Well-formed Types and Terms (implicit context)}
  \label{fig:wftp}
\end{figure}  
While variables might occur free in a type given via the BNF grammar,
the two-dimensional implicit-context formulation models more cleanly
the \emph{scope} of variables; e.g., a type $\istp({\all{\alpha}\fsp\, \alpha\,
  \beta})$ is only meaningful in the context where we have the
assumption $\istp{\beta}$.

Following this judgmental view, we can also characterize well-formed terms
: the rule for term application ($tm_a$) is straightforward and
the rule for type application ($tm_{ta}$) simply refers to the previous judgment for
well-formed types since types are embedded in terms. The rules for term
abstraction ($tm_l$) and type abstraction ($tm_{tl}$) are again the most
interesting. The rule $tm_l$ is parametric in the variable $x$ and hypothetical
in the assumption $\istm{x}$; similarly the rule $tm_{tl}$ is parametric in the
type variable $\alpha$ and hypothetical in the assumption $\istp{\alpha}$.

We emphasize that mechanizations of a given object language can use
either one of these two representations, the BNF grammar or the
judgmental implicit context formulation. However, it is important to
understand how to move between these representations and the
trade-offs and consequences involved. For example, if we choose to
support the BNF-style representation of object languages in a proof
assistant, we
 might need to provide basic predicates that verify whether
a given object is closed; further we may need to reason explicitly
about the scope of variables. HOAS-style proof assistants typically
adopt the judgmental view providing a uniform treatment for objects
themselves (well-formedness rules) and other inference rules about them.

\subsection{Context Definitions}
\label{subsec:ctxdef}
Introducing the appropriate assumption about each variable is a
general methodology that scales to  OLs accommodating much
more expressive assumptions. For example, when we specify
typing rules, we introduce a typing assumption that keeps track of the
fact that a given variable has a certain type. This approach can also result in 
compact and elegant proofs. Yet, it is often convenient to
present hypothetical judgments in a  \emph{localized} form,
reducing some of the ambiguity of the two-dimensional notation. We
therefore introduce an \emph{explicit} context for bookkeeping, since
when establishing properties about a given system, it allows us to
consider the variable case(s) separately and to state clearly when
considering closed objects, i.e., an object in the empty context. More
importantly, while structural properties of contexts are implicitly
present in the above presentation of inference rules (where
assumptions are managed informally), the explicit context presentation
makes them more apparent and highlights their use in reasoning about
contexts.  To contrast representation using explicit contexts to
implicit ones and to highlight the differences, we re-formulate the
earlier rules for well-formed types and terms given in
Fig.~\ref{fig:wftp} using explicit contexts in
Sect.~\ref{ssec:poly-expl}. As another example of using explicit
contexts, we give the standard typing rules for the polymorphic
lambda-calculus (see Sect.~\ref{ssec:poly-expl}). The reader might
want to skip ahead to get an intuition of what explicit contexts are
and how they are used in practice.
In the rest of this section, we first introduce terminology for
structuring such contexts, and then describe structural properties
they (might) satisfy.

Traditionally, a context of assumptions is characterized as a sequence
of formulas $A_1, A_2, \ldots, A_n$ listing its elements separated by
commas~\cite{PierceFirstBook,GirardLafontTaylor:proofsAndTypes}.
However, we argue that this is not expressive enough to capture the
structure  present in contexts, especially when mechanizing OLs.
In fact, there are two limitations from that point of view.  

First, simply
stating that a context is a sequence of formulas does not characterize
adequately and precisely what assumptions can occur in a context and
in what order.  For example, to characterize a well-formed
{type}, we consider a type in a context $\alphctx$ of type
variables. To characterize a well-formed {term}, we must consider
the term in a context $\axctx$ that may contain type variables
$\alpha$ \emph{and} term variables $x$.
\[\begin{array}[t]{llll}
  \mbox{Context}
& \alphctx & \bnfas & \cdot\mid \alphctx,\istp{\alpha} \\
& \axctx & \bnfas & \cdot\mid \axctx,\istp{\alpha} \mid \axctx, \istm{x} \\
\end{array}
\]
As a consequence, we need to be able to state in our mechanization
when a given context \emph{satisfies} being a well-formed context $\alphctx$
or $\axctx$. In other words, the grammar for $\alphctx$ and $\axctx$
will give rise to a \emph{schema}, which describes when a context is
meaningful. Simply stating that a context is a sequence of assumptions does
not allow us necessarily to distinguish between different contexts.

Second, forming new contexts by a {comma} does not capture enough
structure.  For example, consider the typing rule for
lambda-abstraction that states that $\lamt{x}{C}{M}$ has type $(\fsp\,C\,
B)$, if assuming that $x$ is a term variable and $x$ has type $C$, we can
show that $M$ has type $B$. Note that whenever we introduce
assumptions $x \oftp C$ (read as ``term variable $x$ has type $C$''),
we at the same time introduce the additional assumption that $x$ is a \emph{new}
term variable. This is indeed important, since from it we can derive
the fact that every typing assumption is unique. Simply stating that
the typing context is a list of assumptions $x \oftp C$, as shown
below in the first attempt, fails to capture  that $x$ is a
term variable, distinct from all other term variables. In fact, it
says nothing about $x$. 
\[
\begin{array}{llrcl}
\mbox{Typing context} & \mbox{(attempt 1)} & \Phi & \bnfas & \cdot \mid
\Phi, x\oftp C
\end{array}
\]
The second attempt below also fails, because
the occurrences of the comma have two different meanings. 
\[
\begin{array}{llrcl}
\mbox{Typing context} &  \mbox{(attempt 2)}& \Phi & \bnfas & \cdot \mid
\Phi, \istm{x}, x\oftp C 
\end{array}
\]
The comma
between $\istm{x}, x\oftp C$ indicates that whenever we have an
assumption $\istm{x}$, we also have an assumption $x\oftp C$. These
assumptions come in pairs and form one \emph{block} of assumptions. On
the other hand, the comma between $\Phi$ and $\istm{x},x\oftp C$
indicates that the context $\Phi$ is \emph{extended} by the block
containing assumptions $\istm{x}$ and $x\oftp C$.

Taking into account such blocks leads to the definition of contexts as
\emph{structured sequences}. A context is a sequence of declarations
$D$ where a declaration is a block of individual atomic
assumptions separated by '$;$'. The '$;$' binds tighter
than '$,$'. We treat contexts as ordered, i.e., later assumptions in
the context may depend on earlier ones, but not vice versa---this
 in contrast to viewing contexts as multi-sets.

We thus introduce the following categories:
\[
\begin{array}{lrcl}
&\mbox{Atom} \quad A & & \\
&\mbox{Block of declarations} \quad D & \bnfas & A  \mid D; A\\
&\mbox{Context} \quad \Gamma & \bnfas & \cdot \mid \Gamma , D\\
&\mbox{Schema} \quad S & \bnfas & D_s \mid D_s \alt S
\end{array}
\]
Just as types classify terms, a \emph{schema} will classify meaningful
structured sequences. A schema consists of declarations $D_s$, where
we use the subscript $s$ to indicate that the declaration occurring in
a concrete context having schema $S$ may be an \emph{instance} of $D_s$. We
use $\alt$ to denote the alternatives in a  context schema.

We can declare the schemas corresponding to the previous contexts,
seen as structured sequences, as follows:
\[
\begin{array}{llcl}
& S_{\alpha} & \bnfas & \istp{\alpha} \\ 
& \axsch & \bnfas & \istp{\alpha} \alt \istm{x} \\ 
& \atpsch & \bnfas & \istp{\alpha} \alt \istm{x};x\oftp C 
\end{array}
\]

We use the following notational convention for declarations and
schemas: Lower case letters denote bound variables (eigenvariables),
obeying the Barendregt variable convention;
$\FV(D)$ will denote the set of eigenvariables occurring in $D$.
Upper case letters are used for ``schematic'' variables. Therefore, we
can always rename the $x$ in the declaration $\istm{x};x\oftp C$ and
instantiate $C$.
For example, the context $\istm{y};y\oftp\;\mathsf{nat},\;
\istp\alpha,\;\istm{z};z\oftp\;(\mathsf{\fsp\, \alpha\, \alpha})$ fits
the schema $\atpsch$.
\footnote{Although a schema does not appear to have an explicit
  binder, all the eigenvariables and schematic variables occurring are
  considered bound. In ORBI (see Sect.~\ref{sec:mech}) the
  \lstinline{block} keyword delineates the scope of eigenvariables and
  we use the convention that schematic variables are written using
  upper case letters. Beluga's type theory provides a formal
  type-theoretic foundation for describing schemas where the scope of
  eigenvariables and schematic variables in a schema is enforced using
  $\Sigma$ and $\Pi$-types.}

We say that a declaration $D$ is \emph{well-formed} if for every
$x\in\FV(D)$ there is an atom in $D$ (notation $\memb{A}{D}$) denoting
the well-formedness judgment for $x$, which we generically refer to as
$\wf{x}$, with the proviso that $\wf{x}$ precedes its use in $D$; the
meta-notation $\wf\!\!$ will be instantiated by an appropriate atom
such as $\mathsf{is\_tm}$ or $\mathsf{is\_tp}$.  A schema is
\emph{well-formed} if and only if all its declarations are well-formed. For
example, the schema $\atpsch$ is well-formed since the $x$ in $x\oftp
C$ is declared by $\istm x$ appearing earlier in the same declaration.
We will assume in the following that all schemas are
such. %

More generally, we say that a concrete context $\Gamma$ \emph{has
  schema} $S$ ($\schema{\Gamma}{S}$), if every declaration in $\Gamma$
is an instance of some schema declaration $D_s$ in $S$.
By convention, when we write $S_l$ to denote a context schema,
$\Gamma_l$ will denote a valid instance of $S_l$, namely such that
$\schema{\Gamma_l}{S_l}$, where subscript $l$ is used to denote the
relationship between the schema and an instance of it.

\[
\begin{array}{l}
\mbox{Schema Satisfaction }\qquad \fbox{$\schema{\Gamma}{S}$}\\[1em]
\infer{\schema{\;\;\cdot\;\;}{S}}{}
\qquad
\infer{\schema{(\Gamma, D)}{S}}{\schema{\Gamma}{S} \qquad \regdecl{D}{S} \qquad
  \FV(D) \cap \FV{(\Gamma) = \emptyset}} 
\\[1em] 
\mbox{Block $D$ of Declaration is valid }\qquad\fbox{$\regdecl{D}{S}$}\\[1em]
\infer{\regdecl{D}{D_s}}{D \ \mbox{instance of} \ D_s} \qquad
\infer{\regdecl{D}{D_s \alt S}}{D \ \mbox{instance of}\  D_s}\qquad
\infer{\regdecl{D}{D_s \alt S}}{\regdecl{D}{S}} 
\end{array}
\]

Note that if $D \in S$, then it is by definition
well-formed. The premise $\FV(D) \cap \FV(\Gamma) = \emptyset$ requires
eigenvariables in different blocks in a context satisfying the schema
to be distinct from each other.  This constraint will always be
satisfied by contexts that appear in proofs of judgments using our
inference rules---again, see for example the inference rules in
Sect.~\ref{ssec:poly-expl}.
We remark that a given context can in principle inhabit different
schemas; for example the context $\istp{\alpha_1}, \istp{\alpha_2}$
has schema $S_{\alpha}$ but also inhabits schemas $\axsch$ and
$\atpsch$.

Note that according to the given grammar for schemas, contexts contain
only atomic assumptions. We could consider non-atomic assumptions; in
fact, more complex assumptions are not only possible, but sometimes
yield very compact and elegant specifications, as we touch upon in
Sect.~\ref{sec:concl}. However, to account for them, we would need to
introduce a language for terms and formulas that we feel would detract
from the goal at hand.

\subsection{Structural Properties of Contexts}
 \label{subsec:struct}

So far we have introduced terminology for describing objects in three
different ways: using a BNF grammar, 
defining objects and rules via a two-dimensional implicit context, and
using an explicit context containing structured sequences of
assumptions following a given context schema.  For the latter, we have
not yet described the associated inference rules.  Before we do (in
Sect.~\ref{ssec:poly-expl} as mentioned), we introduce structural
properties of explicit contexts \emph{generically} and
\emph{abstractly}.

We concentrate here on developing a common framework for {describing}
object languages including structural properties they might
satisfy. However, we emphasize that whether a given object language
does admit structural properties such as weakening or exchange is a
property that needs to be verified on a case-by-case basis. In the
subsequent discussion and in all our benchmarks, we concentrate on
examples satisfying weakening, exchange, and strengthening, i.e.,
assumptions can be used as often as needed, they can be used in any
order, and certain assumptions will be known not to be needed.

Our refined notion of context has an impact on structural properties
of contexts: e.g., weakening can be described by adding a new
declaration to a context, as well as adding an element inside a block
of declarations.  We distinguish between structural properties of a
\emph{concrete} context and structural properties of \emph{all}
contexts of a given schema. For example, given the context schemas
$S_{\alpha}$ and $\axsch$, we know that all concrete contexts of
schema $\axsch$ can be \emph{strengthened} to obtain a concrete
context of schema $S_\alpha$. Dually, we can think of \emph{weakening}
a context of schema $S_\alpha$ to a context of schema $\axsch$.  We
introduce the operations $\mathsf{rm}$ and $\mathsf{perm}$,
where
 \textsf{rm} removes an element of a declaration, and
\textsf{perm} permutes the elements within a declaration.
 
\begin{define}[Operations on Declarations]
  \begin{itemize}
  \item Let $\mathsf{rm}_A:S \rightarrow S'$ be a total function
    taking a (well-formed) declaration $\regdecl{D}{S}$ and returning
    a (well formed) declaration $\regdecl{D'}{S'}$ where $D'$ is $D$
    with $A$ removed, if $\memb{A}{D}$; otherwise $D'=D$.
  \item Let ${\mathsf{perm}_{\pi}}:S \rightarrow S'$ be a total
    function that permutes the elements of a (well-formed)
    declaration $\regdecl{D}{S}$ according to $\pi$ to obtain a (well
    formed) declaration $\regdecl{D'}{S'}$.
  \end{itemize}
\label{def:struct}
\end{define}

Using these operations on declarations we  state structural
{properties} of declarations, later to be extended to
contexts. These make no assumptions and give no guarantees about the
schema of the context $\Gamma, D$ and the resulting context
$\Gamma, f(D)$ where $f\in
\{\mathsf{rm}_A,\,\mathsf{perm}_{\pi}\}$. In fact,  we often want to
use these properties when $\Gamma$ satisfies some schema $S$, but $D$
does not \emph{yet} fit $S$; in this case, we apply an operation to $D$
so that $\Gamma, f(D)$ \emph{does} satisfy the schema $S$.

Since our context schema may contain 
alternatives, the function \textsf{rm} is defined via
case-analysis covering all the possibilities, where
we describe dropping all assumptions of a case using a dot, e.g.,
$\istm{x}\mapsto \cdot$.  For example:

\begin{itemize}
\item $\mathsf{rm}_{x\oftp A} : \atpsch \rightarrow \axsch = \lambda
  d. \caseof{d}{\istp{\alpha} \mapsto \istp{\alpha} \mid
    \istm{y};y\oftp A \mapsto \istm{y}}$
\item $\mathsf{rm}_{\istm{x}} : \axsch \rightarrow \alphsch = \lambda
  d. \caseof{d}{\istp{\alpha} \mapsto \istp{\alpha} \mid \istm{y}
    \mapsto \cdot}$
\end{itemize}

\begin{prop}[Structural Properties of Declarations] 
  \begin{enumerate}
  \item Declaration Weakening: 
\[
\begin{array}{c}
\infer[\textit{d-wk}]{\Gamma, D, \Gamma' \vdash J}{\Gamma, \mathsf{rm}_A(D),\Gamma' \vdash J}
\end{array}
\]

\item Declaration Strengthening:
\[
\begin{array}{c}
\infer[\textit{d-str}\dagger]{\Gamma, \mathsf{rm}_A(D), \Gamma' \vdash J}{\Gamma,D,\Gamma' \vdash J}
\end{array}
\]
with the proviso $(\dagger)$ that $A$ is irrelevant to
$J$ and $\Gamma'$.\footnote{In practice, this may be done by maintaining a
  dependency call graph of all judgments.}

\item Declaration Exchange:
\[
\begin{array}{c}
\infer[\textit{d-exc}]{\Gamma, \mathsf{perm}_\pi(D), \Gamma' \vdash J}{\Gamma, D,\Gamma'  \vdash J}  

\end{array}
\]


  \end{enumerate}
\label{prop:structd}
\end{prop}

The special case  $\mathsf{rm}_A(A)$  drops $A$ completely, since 
\[
\mathsf{rm}_A = \lambda d. \caseof{d}{A \mapsto \cdot \mid \ldots}
\] 
We treat $\Gamma, \cdot, \Gamma'$ as equivalent to $\Gamma,
\Gamma'$. Hence, in the special case where we have
$\Gamma,\mathsf{rm}_A(A),\Gamma'$, we obtain the well-known weakening
and strengthening laws on contexts that are often stated as:
\[
\begin{array}{c}
\infer[str\dagger]{\Gamma,\Gamma' \vdash J}{\Gamma, A, \Gamma' \vdash J} \qquad
\infer[wk]{\Gamma,A, \Gamma' \vdash J}{\Gamma,  \Gamma' \vdash J} 
\end{array}
\]
In contrast to the above, 
the general exchange property on blocks of declarations cannot be
obtained ``for free'' from the above operations and we define it explicitly:
\begin{prop}[Exchange]
\[
\begin{array}{c}
\infer[exc]{\Gamma, D, D', \Gamma' \vdash J}{\Gamma, D', D, \Gamma' \vdash J}  
\end{array}
\]
with the proviso that the sub-context $D, D'$ is well-formed.
\end{prop}

Further, we state structural properties of \emph{contexts} generically. To
``strengthen'' \emph{all} declarations in a given context $\Gamma$, we simply
write $\mathsf{rm}_A^*(\Gamma)$ using the $*$ superscript.
More generally, by $f^*$ with $f \in
\{\mathsf{rm}_A,\;\mathsf{perm}_\pi\}$, we mean the \emph{iteration}
of the operation $f$ over a context.

\begin{prop}[Structural Properties of Contexts] 
  \begin{enumerate}
  \item Context weakening
\[
\begin{array}{c}
\infer[\textit{c-wk}]{\Gamma\vdash J}{\mathsf{rm}_A^*(\Gamma) \vdash J}  
\end{array}
\]

  \item Context strengthening
\[
\begin{array}{c}
\infer[\textit{c-str}\dagger]{\mathsf{rm}_A^*(\Gamma)\vdash J}{\Gamma \vdash J}  
\end{array}
\]
with the proviso $(\dagger)$ that declarations that are instances of
$A$ are irrelevant to $J$.

  \item Context exchange
\[
\begin{array}{c}
\infer[\textit{c-exc}
]{\mathsf{perm}_\pi^*(\Gamma) \vdash J}{\Gamma\vdash J}  
\end{array}
\]

  \end{enumerate}

\label{prop:structc}
\end{prop}

Finally, by $\mathsf{rm}_D$ (resp.\ $\mathsf{rm}^*_D$), we mean the
iteration of $\mathsf{rm}_A$ (resp.\ $\mathsf{rm}^*_A$) for every
$\memb A D$, while keeping the resulting declaration and the overall
context well-formed,
e.g.\ $\mathsf{rm}_{\istm{y};\,y\oftp A} (\_) =
\mathsf{rm}_{\istm{y}}(\mathsf{rm}_{y\oftp A} (\_)) $. All the above
properties are \emph{admissible} with respect to those extended $\mathsf{rm}$
functions. 

The following examples  illustrate some of the subtleties of this
machinery:

\begin{itemize}
\item $\Gamma, \mathsf{rm}_{x\oftp A}(\istm{y};y\oftp A) = \Gamma,
  \istm{y}$. Bound variables in the annotation of $\mathsf{rm}$ can
  always be renamed so that they are consistent with the
  eigenvariables used in the declaration.

\item $\mathsf{rm}^*_{\istm{x}}(\istm{x_1}, \istp{\alpha},
  \istp{\beta}, \istm{x_2}) = \istp{\alpha}, \istp{\beta}$.
  Here, the $\mathsf{rm}$ operation drops one of the alternatives in the
  schema $\axsch$.
\item $\mathsf{rm}^*_{y\oftp A}(\istm{x_1};x_1\oftp \mathsf{nat},\;\istm{x_2};x_2\oftp \mathsf{bool},\;\istp{\alpha}) = 
(\istm{x_1}, \istm{x_2}, \istp{\alpha})$. The schematic variable $A$ occurring in the annotation of
$\mathsf{rm}$ will be instantiated with \textsf{nat} when strengthening the block  
$\istm{x_1};x_1\oftp \mathsf{nat}$ and similarly with \textsf{bool}.

\item $\mathsf{rm}^*_{\istm{y};\,y\oftp A} (\istp{\alpha},\;\istp{\beta}) =
  (\istp{\alpha},\;\istp{\beta}) $. 
A $\mathsf{rm}$ operation may leave a context unchanged.
\end{itemize}

We state next the substitution properties for assumptions.  The
\emph{parametric substitution} property allows us to instantiate
parameters, i.e., eigenvariables, in the context. For example, given
$\istp{\alpha},\istp{\beta} \vdash J$ and a type $\mathsf{bool}$, we
can obtain $\istp{\mathsf{bool}},\istp{\beta} \vdash
[\mathsf{bool}/\alpha]J$ by replacing $\alpha$ with
$\mathsf{bool}$. 
The \emph{hypothetical
  substitution} property allows us to eliminate an atomic formula $A$
that is part of a declaration $D$. 
For example, given $\istp{\mathsf{bool}},\;\istp{\beta} \vdash J$ and
evidence that $\istp{\mathsf{bool}}$, we can obtain $\istp{\beta}
\vdash J$. In type theory the two substitution properties collapse
into one.

\begin{prop}[Substitution Properties]
  \begin{itemize}
  \item \textbf{Hypothetical Substitution}: \\If $\Gamma_1, (D_1;A;D_2), \Gamma_2 \vdash J$
    and $\Gamma_1, D_1 \vdash A$, then $\Gamma_1,(D_1;D_2), \Gamma_2 \vdash J$
    provided that $D_1;D_2$ is a well-formed declaration in $\Gamma_1$.
  \item \textbf{Parametric Substitution}: \\ If $\Gamma_1,(D_1;\wf{x};D_2), \Gamma_2 \vdash J$,
    then $\Gamma_1,(D_1; [t/x]D_2),\,[t/x]\Gamma_2  \vdash
    [t/x]J$ for any term $t$ for which $\Gamma_1, D_1 \vdash \wf{t}$
    holds.
  \end{itemize}
\end{prop}

While parametric and hypothetical substitution do not preserve schema
satisfaction by definition, 
we typically use them in such a way that contexts
continue to satisfy a given schema.

We close this section recalling that, although we concentrate in our
benchmarks on describing object languages that satisfy structural
properties usually associated with intuitionistic logic,
we note that our terminology can be used to also characterize sub-structural
object languages. In the case of a linear object language, we might
choose to only use operations such as $\mathsf{perm}$ and omit
operations such as $\mathsf{rm}$ so as to faithfully and adequately
characterize the allowed context operations.

\subsection{The Polymorphic Lambda-Calculus Revisited}
\label{ssec:poly-expl}

In systems supporting HOAS, inference rules are usually expressed
using an implicit-context representation as illustrated in
Fig.~\ref{fig:wftp}.  The need for explicit structured contexts, as
discussed in Sects.~\ref{subsec:ctxdef} and~\ref{subsec:struct},
arises when performing meta-reasoning about the judgments expressed by
these inference rules.  In order to make the link, we revisit the
example from Sect.~\ref{subsec:theory} giving a presentation with
explicit contexts, and then we make some preliminary remarks about
context schemas and meta-reasoning.  We will adopt the
explicit-context representation of inference rules in the rest of the
paper with the informal understanding of how to move between the
implicit and explicit formulations.

\[
\begin{array}{l}
\multicolumn{1}{l}{\mbox{Well-formed Types}} \\[1em]  
\infer[tp_{v}]{\Gamma \vdash \istp{\alpha }}{\istp{\alpha } \in \Gamma} \quad\quad
\infer[tp_{ar}]{\Gamma \vdash \istp{(\fsp A B)}}{\Gamma \vdash \istp{A} & 
\Gamma \vdash \istp{B}} \qquad\qquad
\infer[tp_{al}]{\Gamma \vdash \istp{(\all{\alpha}{A})}}{
  \Gamma, \istp{\alpha} \vdash \istp{A}
}\\[1em]

\multicolumn{1}{l}{\mbox{Well-formed Terms}} \\[1em]  
\infer[tm_{v}]{\Gamma \vdash \istm{x}}{\istm{x} \in \Gamma}
\quad
\infer[tm_{l}]{\Gamma \vdash \istm{(\lamt{x}{A}{M}})}{
& 
\Gamma, \istm{x} \vdash \istm{M}}
\qquad
\infer[tm_{tl}]{\Gamma \vdash \istm{(\tlam{\alpha}{M}})}{\Gamma, \istp{\alpha} \vdash
  \istm{M}}
\\[1em]
 \infer[tm_a]{\Gamma \vdash \istm{(\app{M_1}{M_2})}}
 {\Gamma \vdash \istm{M_1}\qquad  \Gamma \vdash \istm{M_2}
 }
\quad\infer[tm_{ta}]{\Gamma \vdash \istm{(\tapp{M}{A})}}
  {\Gamma \vdash \istm{M}\qquad\Gamma\vdash \istp{A}
  } \\[1em]
\multicolumn{1}{l}{\mbox{Typing for the Polymorphic $\lambda$-Calculus}}\\[1em]
 \ianc{x\oftp B\in \Gamma}{\Gamma \vd x \hastype B}{{\mathit{of}_v}}  
\quad 
\ianc{\Gamma, \istp{\alpha} \vd M \hastype B}{\Gamma \vd
  \tlam{\alpha}{M} \hastype 
  \all \alpha B}{{\mathit{of}_{tl}}}    
\quad
\ianc{\Gamma \vd M \hastype \all \alpha B \quad \Gamma \vd \istp{B}}{\Gamma \vd (\tapp M   B)
   \hastype [B/\alpha]A}{{\mathit{of}_{ta}}}
\vsk
\ianc{\Gamma, \istm{x};x\oftp A \vd M \hastype B}{\Gamma \vd \lamt x A M \hastype
  \fsp A B}{{\mathit{of}_l}}   
\qquad
\ibnc{\Gamma \vd M \hastype \fsp B A}{\Gamma \vd N \hastype B}{\Gamma \vd (\app M   N)
   \hastype A}{{\mathit{of}_a}}
 \end{array}
\]

In this formulation, and differently from the implicit one, we have a
base case for variables. Here, to look up an assumption in a context,
we simply write $A \in \Gamma$, meaning that there is some block $D$
in context $\Gamma$ such that $\memb A D$.  For example $x\oftp B \in
\Gamma$ holds if $\Gamma$ contains block $\istm{x}; x\oftp B$.  We
will also overload the notation and write $D\in\Gamma$ to indicate
that $\Gamma$ contains the entire block $D$.  We recall the
distinction between the comma used to separate blocks, and the
semi-colon used to separate atoms within blocks, as seen in the
$\mathit{of}_l$ rule, for example.  The assumption that all variables
occurring in contexts are distinct from one another is silently
preserved by the implicit proviso in rules that extend the context,
where we rename the bound variable if already present.

Note that we use a generic $\Gamma$ for the context appearing in these
rules, whereas the reader may have expected this to be, for example,
${\atpctx}$ having schema ${\atpsch}$ in the typing rules.
In fact, we take a more liberal approach, where we pass to the rules
\emph{any} context that can be seen as a \emph{weakening}
of $\atpctx$; in other words, any $\Gamma$ such that there exists a
$D$ for which $\mathsf{rm}^*_D(\Gamma) = \atpctx$.

Suppose now, to fix ideas, that $\atpctx \vd M \hastype B$ holds.  By
convention, we implicitly assume that both $B$ and $M$ are
well-formed, which means that $\atpctx \vd \istp{B}$ and $\atpctx \vd
\istm{M}$.  In fact, we can define functions $\mathsf{rm}^*_{x\oftp
  C}$ and $\mathsf{rm}^*_{\istm{x};x\oftp C}$, use them to define
strengthened contexts $\axctx$ and $\alphctx$, and apply the
\textit{c-str} rule to conclude the following:
\begin{enumerate}
\item $\axctx := \mathsf{rm}^*_{x\oftp C}(\atpctx)$, \quad
  $\schema{\axctx}{\axsch}$, \quad and $\axctx \vdash \istm{M}$;
\item $\alphctx := \mathsf{rm}^*_{\istm{x};x\oftp C}(\atpctx)$, \quad
  $\schema{\alphctx}{\alphsch}$, \quad and $\alphctx \vdash \istp{B}$.
\end{enumerate}

\subsection{Generalized Contexts vs.\ Context Relations}
As an alternative to using functions such as $\mathsf{rm}^*_{x\oftp C}$
in item (1), we may adopt the more suggestive notation
$\crel{\axctx}{\atpctx}$, using inference rules for the context
relation corresponding to the graph of the function $\lambda
d. \caseof{d}{\istp{\alpha} \mapsto \istp{\alpha} \mid \istm{x};x\oftp
  C \mapsto \istm{x}}$:

\[
\begin{array}{c}
\infer{\crel \cdot \cdot}{}\qquad
\infer{
\crel {(\axctx, \istp{\alpha})} {(\atpctx, \istp{\alpha})}
}{\crel \axctx \atpctx} \qquad
\infer{
\crel {(\axctx, \istm{x})} {(\atpctx, \istm{x};x\oftp B)}
}{\crel \axctx \atpctx}\end{array}
\]
Similarly, an alternative to $\mathsf{rm}^*_{\istm{x};x\oftp C}$ in
item (2) is
the following context relation:
\[
\begin{array}{c}
\infer{\crel \cdot \cdot}{}\qquad
\infer{
\crel {(\alphctx, \istp{\alpha})} {(\atpctx, \istp{\alpha})}
}{\crel \alphctx \atpctx} \qquad
\infer{
\crel {\alphctx} {(\atpctx, \istm{x};x\oftp B)}
}{\crel \alphctx \atpctx}
\end{array}
\]
%
The above two statements can now be restated using these relations.
Given $\atpctx$, let $\axctx$ and $\alphctx$ be the unique contexts
such that:
\begin{enumerate}
\item $\crel{\axctx}{\atpctx}$, \quad
  $\schema{\axctx}{\axsch}$, \quad and $\axctx \vdash \istm{M}$;
\item $\crel{\alphctx}{\atpctx}$, \quad
  $\schema{\alphctx}{\alphsch}$, \quad and $\alphctx \vdash \istp{B}$.
\end{enumerate}

When stating and proving properties, we often relate two judgments 
to each other, where each one has its own contexts.  For
example, we may want to prove statements such as ``if $\axctx \vdash
J_1$ then $\atctx \vdash J_2$.'' The question is how  we achieve
that.  In the benchmarks in this paper, we consider two approaches:
\begin{enumerate}
\item We reinterpret the statement in the \emph{smallest} context that
  collects all relevant assumptions;  we call this the
  \emph{generalized context} approach (G). In this case, we
  reinterpret the above statement about $J_1$ in a  context containing
  additional assumptions about typing, which in this case is
  $\atpctx$, yielding:
  \begin{center}
    ``if $\atctx \vdash J_1$ then $\atctx \vdash J_2$.''
  \end{center}

\item We state how two (or more) contexts are \emph{related}; we call
  this the \emph{context relations} approach (R).  Here,
  we define context relations such as those above and use them
  \emph{explicitly} in the statements of theorems.  In this case, we
  use $\crel \axctx \atctx$ yielding:
  \begin{center}
   ``if $\axctx \vdash J_1$ and $\crel \axctx
  \atctx$ then $\atctx \vdash J_2$.''
  \end{center}
Note that here too we ``minimize'' the relations, in the sense of
relating the smallest possible contexts where the relevant judgments
make sense.
\end{enumerate}

\subsection{Context Promotion and Linear Extension of Contexts and Schemas}
Another common idiom in meta-reasoning occurs when we have established
a property for a particular context and we would like to use this
property subsequently in a more general context.  Assume that we have
proven a lemma about types in context $\alphctx$ of the form ``if
$\alphctx \vdash J_1$ then $\alphctx \vdash J_2$.'' We now want to
\emph{use} this lemma in a proof about terms, that is where we have a
context $\axctx$ and $\axctx \vdash J_1$.  We may need to \emph{promote}
this lemma, and prove: ``if $\axctx \vdash J_1$ then $\axctx \vdash
J_2$.''  We will see several examples of such promotion lemmas in
Sect.~\ref{sec:bench}.

Finally, to structure our subsequent discussion, it is
useful to introduce some additional terminology regarding context
relationships, where we use ``relationship'' in contrast to the more
specific notion of ``context relation.''
\begin{itemize}
\item \emph{Linear extension of a declaration}: a declaration $D_2$ is
  a \emph{linear extension} of a declaration $D_1$, if every atom in the
  declaration $D_1$ is a member of the declaration $D_2$.
\item \emph{Linear extension of a schema}: a schema $S_2$ is a
  \emph{linear extension} of a schema $S_1$, if every declaration 
  in $S_1$ is a linear extension of a declaration in $S_2$. For
  example $\atpsch$ is a linear extension of $\axsch$. 
\end{itemize}
Given a context $\Phi_1$ of schema $S_1$ and a context $\Phi_2$ of
schema $S_2$ where $S_2$ is a linear extension of $S_1$, we say that
$\Phi_2$ is a linear extension of $\Phi_1$ (i.e., linear context
extension).  Of course,  sometimes declarations, schemas and contexts are not
related linearly. For example, we may have a schema $S_2$ and a schema
$S_3$ both of which are linear extensions of $S_1$; however, $S_2$ is
not a linear extension of $S_3$ (or vice versa). In this case, we say
$S_2$ and $S_3$ are \emph{non-linear extensions} of each other and
they share a most specific common fragment.

\section{Benchmarks}
\label{sec:bench}

In this section, we present several case studies establishing proofs of
various properties of the lambda-calculus. We
have structured this section around the different shapes and
properties of contexts, namely:
\begin{enumerate}
\item Basic linear context extensions: We consider here contexts
  containing no alternatives. We refer to such contexts as
  \emph{basic}. We discuss context membership and revisit structural
  properties such as weakening and strengthening.
\item Linear context extensions with alternative declarations. %
\item 
  Non-linear context extensions: We consider more complex
  relationships between contexts 
   and discuss how our proofs involving weakening and
  strengthening change.
\item Order: 
  We consider how the ordered structure of contexts impacts proofs
  relying on exchange.
\item Uniqueness: We consider here a case
  study which highlights how the issue of distinctness of all variable
  declarations in a context arises in proofs.
\item Substitution: Finally, we exhibit  the fundamental properties of
hypothetical and  parametric substitution.
\end{enumerate}

The benchmark problems are purposefully \emph{simple}; they
are designed to be easily understood so that one can quickly
appreciate the capabilities and trade-offs of the different systems in
which they can be implemented. Yet we believe they are representative
of the issues and problems arising when encoding formal systems and
reasoning about them. We will subsequently discuss both the G approach
and the R approach and comment on the trade-offs and differences in
proofs depending on the chosen approach.

\subsection{Basic Linear Context Extension}
\label{ssec:basic}
We concentrate in this section on contexts with simple schemas consisting
of a single declaration. We aim to show the basic building blocks of reasoning
over open terms: namely what a context looks like and the structure of an
inductive proof. For the latter, we focus on the case analysis and,
at the risk of being pedantic, the precise way in which the induction
hypothesis is applied.

We start with a very simple judgment: \emph{algorithmic equality} for
the untyped lambda-calculus, written $(\termeq{M}{N}$), also known as
\emph{copy} clauses~\citep[see][]{Miller91jlc}.
We say that two terms are algorithmically equal
provided they have the same structure with respect to the
constructors. 
\[
\begin{array}{c}
\multicolumn{1}{l}{\mbox{Algorithmic Equality}}\\[0.75em]
\infer[ae_v]{\Gamma \vdash \termeq{x}{x}}{\termeq{x}{x} \in \Gamma}
\qquad 
\infer[ae_l]{\Gamma \vdash \termeq{(\lam{x}{M})}{(\lam{x}{N})}}{
\Gamma, \istm{x}; \termeq{x}{x} \vdash \termeq{M}{N}} 
\\[+1em]
\infer[ae_a]{\Gamma \vdash \termeq{(\app{M_1}{M_2})}{(\app{N_1}{N_2})}}
{\Gamma \vdash \termeq{M_1}{N_1} & \Gamma \vdash \termeq{M_2}{N_2}
}
\end{array}
\]
The context schemas needed for reasoning about this judgment are the
following:
\[
 \label{page:ctxdef}
\begin{array}[t]{llll}
\mbox{Context Schemas}
               & \isch & := & \istm{x} \\
               & \asch & := & \istm{x}; \termeq{x}{x} \\
\end{array}
\]
where a context $\actx$ satisfying $\asch$ is the smallest possible
context in which such an equality judgment can hold.
Thus, as discussed in the previous section,  when writing judgment $\actx
\vdash \termeq{M}{N}$, we assume that 
$\actx \vdash\istm{M}$ and $\actx \vdash\istm{N}$ hold, 
and thus also $\ictx \vdash\istm{M}$ and $\ictx \vdash\istm{N}$ hold
by employing an implicit \copda{str}{decl} (using
$\mathsf{rm}_{\termeq{x}{x}}^*$). We note that both contexts $\ictx$ and $\actx$
are simple contexts consisting of one declaration block. Moreover,
$\isch$ is a sub-schema of $\asch$ and therefore the context $\actx$
is a linear extension of the context $\ictx$.

In view of the pedagogical nature of this subsection and also of the
content of Sect.~\ref{ssec:sub}, which will build on this example,
we start with a straightforward property: algorithmic equality is reflexive.
This property should follow by induction on $M$ (via the well-formed
term judgment, which is not shown, but uses the obvious subset of the
rules in Sect.~\ref{ssec:poly-expl}).
However, the question of which contexts the two judgments should be stated in
arises immediately; recall that we want to prove ``if $\Gamma_1 \vdash \istm{M}$
then $\Gamma_2 \vdash \termeq{M}{M}$.''  $\Gamma_2$ should be a
context satisfying $\asch$ since the definition of this schema came
directly from the inference rules of this judgment.  The form that
$\Gamma_1$ should take is less clear.  The main requirement comes from
the base case, where we must know that for every assumption $\istm{x}$
in $\Gamma_1$ there exists a corresponding assumption $\termeq{x}{x}$
in $\Gamma_2$.  The answer differs depending on whether we choose the
R approach or the G approach.  We discuss each in turn below.

\subsubsection{Context Relations, R Version}

\label{sssec:rel}
The relation needed here is $\crel \ictx \actx$, defined as follows:
\[
\begin{array}{c}
  \multicolumn{1}{l}{\mbox{Context Relation}}\\[0.75em]
\infer[crel_e]{\crel \ldot \ldot}{} \quad\quad\quad\quad
\infer[crel_{xa}]{\crel {\ictx, \istm x} {\actx, \istm{x}; \termeq{x}{x}}}{\crel \ictx \actx}
\end{array}
\]
Note that
$\istm x$ will occur in $\ictx$ in sync with an assumption block
containing $\istm{x}; \termeq x x$ in $\actx$. This is a
property which needs to be established separately, so at the risk of
redundancy, we state it as a ``member'' lemma.
\begin{lem}[Context Membership]
\label{le:member_refl}
   $\crel \ictx \actx$ implies that $\istm x\in\ictx$
iff $\istm x;\termeq x x\in\actx$.
\end{lem}
\begin{proof}
  By induction on $\crel \ictx \actx$.
\end{proof}
\begin{theorem}[Admissibility of Reflexivity, R Version]
\label{thm:reflR}
  Assume $\crel \ictx \actx$. \\
If $\ictx \vdash \istm{M}$ then $\actx \vdash \termeq{M}{M}$.
\end{theorem}
\begin{proof}
By induction on the derivation $\cD :: \ictx \vdash \istm{M}$.

\noindent
Case: $$\cD = \ianc{\istm{x} \in \ictx}{\ictx \vdash \istm{x}}{tm_v}$$
$\istm{x}\in\ictx$ \lcase{by rule premise}
$\istm x;\termeq{x}{x}\in\actx$ \lcase{by Lemma~\ref{le:member_refl}}
$\actx \vdash \termeq{x} x$\lcase{by rule $ae_v$}

\noindent
Case: $$\cD = \ibnc{\aboved {\cD_1}{ \ictx \vdash\istm {M_1}}}{
 \aboved
  {\cD_2}
  { \ictx \vdash\istm {M_2}}} {
\ictx \vdash  \istm{(\app{M_1}{M_2})}}{tm_a}$$
$\ictx \vdash\istm {M_1}$\lcase {sub-derivation $\cD_1$}
$\actx \vdash\termeq {M_1} M_1$\lcase {by IH}
$\ictx \vdash\istm {M_2}$\lcase {sub-derivation $\cD_2$}
$\actx \vdash\termeq {M_2} M_2$\lcase {by IH}
$\actx \vdash \termeq{(\app{M_1}{M_2}) \ (\app{M_1}{M_2})}$\lcase{by
  rule $ae_a$}

\noindent
Case: 
$$\cD = \ianc{\aboved {\cD'}{ \ictx, \istm x \vdash\istm  M}}{\ictx\vd\istm{(\lam x M)}}{tm_l}$$
$\ictx,\istm{x} \vdash\istm{M}$ \hfill
sub-derivation  $\cD'$\\
$\crel \ictx \actx$ \lcase{by assumption}
$\crel {(\ictx,\istm x)} {(\actx,\istm{x};\termeq x x)}$\lcase{by rule $crel_{xa}$}
$\actx, \istm{x}; \termeq{x}{x} \vdash \termeq{M}{M}$ \hfill
by IH \\
$\actx \vdash \termeq{(\lam{x}{M})}{(\lam{x}{M})}$ \hfill by rule
$ae_l$.
\end{proof}

\subsubsection{Generalized Contexts, G Version}
\label{sssec:gen}
In this example, since $\asch$ includes all assumptions in $\isch$,
$\asch$ will serve as the schema of our generalized context.

\begin{theorem}[Admissibility of Reflexivity, G Version]
\label{thm:reflG}
  If $\actx \vdash \istm{M}$ then $\actx \vdash \termeq{M}{M}$.
\end{theorem}

\begin{proof}
By induction on the derivation $\cD :: \actx \vdash \istm{M}$.

\noindent
Case: $$\cD = \ianc{\istm{x} \in \actx}{\actx \vdash \istm{x}}{tm_v}$$
$\istm{x} \in \actx$\lcase{by rule premise}
$\actx$ contains block
 $(\istm x; \termeq{x}{x})$ \lcase{by definition of $\asch$}
$\actx \vdash \termeq{x} x$\lcase{by rule $ae_v$}

\noindent
Case: $$\cD = \ibnc{\aboved {\cD_1}{ \actx \vdash\istm {M_1}}}{
\aboved  {\cD_2} {\actx \vdash\istm {M_2}}} {\actx \vdash
  \istm{(\app{M_1}{M_2})}}{tm_a}$$
$\actx \vdash\termeq {M_1} M_1$\lcase {by IH on $\cD_1$}
$\actx \vdash\termeq {M_2} M_2$\lcase {by IH on $\cD_2$}
$\actx \vdash \termeq{(\app{M_1}{M_2}) \   (\app{M_1}{M_2})}$\lcase{by
  rule $ae_a$} 

\noindent
Case: 
$$\cD = \ianc{\aboved {\cD'} {\actx, \istm x \vdash\istm  M}}{\actx\vd\istm{(\lam x M)}}{tm_l}$$
%
$\actx, \istm{x}; \termeq{x}{x} \vdash\istm{M}$
\hfill by $\dop wk$ on $\cD'$ \\
$\actx,\istm{x}; \termeq{x}{x} \vdash \termeq{M}{M}$ \hfill
by IH \\
$\actx \vdash \termeq{(\lam{x}{M})}{(\lam{x}{M})}$ \hfill
by rule $ae_l$
\end{proof}
Note that the application cases of Theorems~\ref{thm:reflR}
and~\ref{thm:reflG} are the same except for the context used for the
well-formed term judgment.  The lambda case here, on the other hand,
requires an additional weakening step.  In particular, $\dop wk$ is used to add
an atom to form the declaration needed for schema $\asch$.  The
context before applying weakening does not satisfy this schema, and
the induction hypothesis cannot be applied until it does.

We end this subsection, stating the remaining properties needed to
establish that algorithmic equality is indeed a congruence, which we
will prove in Sect.~\ref{ssec:sub}. Since the proof involves
only $\actx$, the two approaches (R \& G) collapse.
\begin{lem}[Context Inversion]
\label{le:member_shape}  
   If $\termeq M N\in\actx$
then $M = N$.
\end{lem}
\begin{proof}
  Induction on $\termeq M N\in\actx$.
\end{proof}

\begin{theorem}[Admissibility of Symmetry and
    Transitivity] \label{le:sym-trans}\mbox{}
  \begin{enumerate}
  \item If $\actx \vdash \termeq{M}{N}$ then $\actx \vdash
    \termeq{N}{M}$.
  \item If $\actx \vdash \termeq{M}{L}$ and $\actx \vdash
    \termeq{L}{N}$ then $\actx \vdash \termeq{M}{N}$.
  \end{enumerate}
\end{theorem}
\begin{proof}
  Induction on the given derivation using Lemma~\ref{le:member_shape}
  in the variable case.
\end{proof}

\subsection{Linear Context Extensions with Alternative Declarations}
\label{ssec:poly}

We extend our algorithmic equality case study to the polymorphic
lambda-calculus, highlighting the situation where judgments induce
context schemas with \emph{alternatives}.
We accordingly add the judgment for \emph{type equality}, $\tpaeq A
B$, noting that the latter can be defined independently of term
equality. In other words $\termeq{M}{N}$ depends on
$\tpaeq{A}{B}$, but not vice-versa.
In addition to $\alphsch$ and $\axsch$ introduced in
Sect.~\ref{sec:theory}, the following new context schemas are also
used here:
\[
\begin{array}[t]{lll}
   \atpeqsch & := & \istp{\alpha};\tpaeq{\alpha}{\alpha}  \\
   \psch & := & \istp{\alpha};\tpaeq{\alpha}{\alpha}~\alt~\istm{x};\termeq{x}{x} 
\end{array}
\]
The rules for the two equality judgments extend those given in
Sect.~\ref{ssec:basic}.  The additional rules are stated below.
\[
\begin{array}{c}
  \multicolumn{1}{l}{\mbox{Algorithmic Equality for the Polymorphic
      Lambda-Calculus}}\\[0.75em]
  \dots\\[0.75em]
  \infer[ae_{tl}]{\Gamma \vdash \termeq{(\tlam{\alpha}{M})}{(\tlam{\alpha}{N})}}{
    \Gamma, \istp{\alpha};\tpaeq{\alpha}{\alpha} \vdash \termeq{M}{N}} 
  \\[+1em]
  \infer[ae_{ta}]{\Gamma \vdash \termeq{(\tapp{M}{A})}{(\tapp{N}{B})}}
  {\Gamma \vdash \termeq{M}{N} &\Gamma \vdash \tpaeq{A}{B}
  }
\\[+1em]
 \infer[at_{\alpha}]{\Gamma \vdash \tpaeq{\alpha}{\alpha}}{\tpaeq{\alpha}{\alpha} \in \Gamma}
  \qquad 
 \\[+1em]
 \infer[at_{al}]{\Gamma \vdash \tpaeq{(\all{\alpha}{A})}{(\all{\alpha}{B})}}{
    \Gamma, \istp{\alpha};\tpaeq{\alpha}{\alpha} \vdash \tpaeq{A}{B}} 
  \qquad 
  \infer[at_a]{\Gamma \vdash \tpaeq{(\fsp{A_1}{A_2})}{(\fsp{B_1}{B_2})}}
  {\Gamma \vdash \tpaeq{A_1}{B_1} & \Gamma \vdash \tpaeq{A_2}{B_2}
  }
\end{array}
\]

We show again the admissibility of reflexivity.  We start with the G
version this time.

\subsubsection{G Version}

We first state and prove the admissibility of reflexivity for types,
which we then use in the proof of admissibility of reflexivity for
terms.  The schema for the generalized context for the former is
$\atpeqsch$ since the statement and proof do not depend on terms.  The
schema for the latter is $\psch$.
\begin{theorem}[Admissibility of Reflexivity for Types, G Version]\\
   If \mbox{}$\atpeqctx \vdash \istp{A}$ then $\atpeqctx \vdash \tpaeq{A}{A}$.
\label{thm:atpreflG}
\end{theorem}
The proof is exactly the same as the proof of Theorem~\ref{thm:reflG}, modulo
replacing $\capp$ and $\clam$ with $\carr$ and $\call$, respectively,
and using the corresponding rules.

As we have already mentioned in Sect.~\ref{sec:theory}, it is often the
case that we need to appeal to a lemma in a context that is different
from the context where it was proved.  A concrete example is the above
lemma, which is stated in context $\atpeqctx$, but is needed in the
proof of the next theorem in the larger context $\pctx$.
To illustrate, we
state and prove the necessary \emph{promotion} lemma here.
\begin{lem}[G-Promotion for Type Reflexivity]\\
   If $\pctx \vdash \istp{A}$ then $\pctx \vdash \tpaeq{A}{A}$.
\label{lem:polypromoteG}
\end{lem}
\begin{proof}\ \\
$\pctx \vdash \istp{A}$\lcase{by assumption}
$\atpeqctx \vdash \istp{A}$\lcase{by $\copda {str} {alt}$}
$\atpeqctx \vdash \tpaeq{A}{A}$\lcase{by Theorem~\ref{thm:atpreflG}}
$\pctx \vdash \tpaeq{A}{A}$\lcase{by $\copda {wk} {alt}$}
\end{proof}
In general, proofs of promotion lemmas require applications of
\copda{str}{alt} and \copda{wk}{alt} which perform a uniform
modification to an entire context.  In contrast, the abstraction cases
in proofs such as the lambda case of Theorem~\ref{thm:reflG} require
\dop{wk} to add atoms to a single declaration.  The particular
function used here is $\textsf{rm}_{\istm{x};\termeq{x}{x}}^*$, which
drops an entire alternative from $\pctx$ to obtain $\atpeqctx$ and
leaves the other alternative unchanged.  
The combination of \copda{str}{} and \copda{wk}{} in proofs of
promotion lemmas is related to
\textit{subsumption}~\citep[see][]{HarperLicata:JFP07}.

Note that we could omit Theorem~\ref{thm:atpreflG} and instead prove
Lemma~\ref{lem:polypromoteG} directly, 
removing the need for a promotion lemma.  For modularity purposes, we
adopt the approach that we state each theorem in the smallest
possible context in which it is valid.  This particular lemma, for
example, will be needed in an even bigger context than $\pctx$ in
Sect.~\ref{ssec:sub}.  In general, we do not want the choice of
context in the statement of a lemma to depend on later theorems whose
proofs require this lemma.  Instead, we choose the smallest context
and state and prove promotion lemmas where needed.

\begin{theorem}[Admissibility of Reflexivity for Terms, G Version]\\
  \label{thm:prelreflG}
 If $\pctx \vdash \istm{M}$ then $\pctx \vdash \termeq{M}{M}$.
\end{theorem}
\begin{proof}
Again, the proof is by induction on the given well-formed term
derivation, in this case $\cD :: \pctx \vdash \istm{M}$, and is
similar to the proof of Theorem~\ref{thm:reflG}.
We show the case for application of terms to
types.

 \noindent
Case:
$$\cD = \ibnc{\aboved{\cD_1}{\pctx \vdash\istm{M}}}
                   {\aboved{\cD_2}{\pctx \vdash\istp{A}}}
                   {\pctx \vdash \istm{(\tapp{M}{A})}}{}$$
$\pctx \vdash\termeq {M} M$\lcase {by IH on $\cD_1$}
$\pctx \vdash\tpaeq {A}{A} $\lcase {by Lemma~\ref{lem:polypromoteG} on
  conclusion of
$\cD_2$}
$\pctx \vdash \termeq{(\tapp{M}{A})}{(\tapp{M}{A})}$\lcase{by
  rule $ae_{ta}$}
\end{proof}

\subsubsection{R Version}
We introduce four context relations $\crel{\alphctx}{\atpeqctx}$,
$\crel{\axctx}{\pctx}$, $\crel{\axctx}{\alphctx}$, and
$\crel{\pctx}{\atpeqctx}$.  We define the first two as follows (where
we omit the inference rules for the base cases).
\[
\begin{array}{c}
  \multicolumn{1}{l}{\mbox{Context Relations}}\\[0.75em]
\infer
                        {\crel{\alphctx, \istp{\alpha}}
                        {\atpeqctx,\istp{\alpha};\tpaeq{\alpha}{\alpha}}}
                  {\crel{\alphctx}{\atpeqctx}}    \\[5pt]
\infer
                {\crel{\axctx, \istm{x}}{\pctx,\istm{x};\termeq{x}{x}}}
                  {\crel{\axctx}{\pctx}}\qquad
  \infer
                     {\crel{\axctx,\istp{\alpha}}
                        {\pctx,\istp{\alpha};\tpaeq{\alpha}{\alpha}}}
                  {\crel{\axctx}{\pctx}}
\end{array}
\]
Note that $\crel{\axctx}{\pctx}$ is the extension of
$\crel{\ictx}{\actx}$ with one additional case for equality for
types.\footnote{Again, we
  remark on our policy to use the smallest contexts possible for
  modularity reasons.  Otherwise, we could have omitted the
  $\crel{\alphctx}{\atpeqctx}$  relation, and
  stated the next theorem using $\crel{\axctx}{\pctx}$.}
We also omit the (obvious) inference rules defining
$\crel{\axctx}{\alphctx}$ and $\crel{\pctx}{\atpeqctx}$, and instead
note that they correspond to the graphs of the following two
functions, respectively, which simply remove one of the two schema
alternatives:
$$\begin{array}{l}
\textsf{rm}^*_{\istm x} =
 \lambda d. \caseof{d}{\istp{\alpha} \mapsto \istp{\alpha} \mid
  \istm{x} \mapsto \cdot} \\
\textsf{rm}^*_{\istm x ; \termeq x x } =
  \lambda d. \caseof{d}{\istp{\alpha};\tpaeq{\alpha}{\alpha} \mapsto
    \istp{\alpha};\tpaeq{\alpha}{\alpha} \mid \istm{x};\termeq{x}{x}
    \mapsto \cdot}
\end{array}$$

We start with the theorem for types again, whose proof is similar to
the R version of the previous example (Theorem~\ref{thm:reflR}) and is
therefore omitted.
\begin{theorem}[Admissibility of Reflexivity for Types, R Version]\mbox{}\\
  \label{thm:prelreflR_tp}
  Let $\crel{\alphctx}{\atpeqctx}$.  If $\alphctx \vdash \istp{A}$ then $\atpeqctx
  \vdash \tpaeq{A}{A}$.
\end{theorem}
\begin{lem}[Relational Strengthening]
Let $\crel{\axctx}{\pctx}$.  Then there exist contexts $\alphctx$ and
$\atpeqctx$ such that $\crel{\axctx}{\alphctx}$,
$\crel{\pctx}{\atpeqctx}$, and
$\crel{\alphctx}{\atpeqctx}$.
\label{lem:strenR}
\end{lem}
\begin{proof}
By induction on the given derivation of $\crel{\axctx}{\pctx}$.
\end{proof}
 We again need a promotion lemma, this time involving the context
 relation.
 \begin{lem}[R-Promotion for Type Reflexivity]\\
   Let $\crel{\axctx}{\pctx}$. If $\axctx \vdash \istp{A}$ then $\pctx
   \vdash \tpaeq{A}{A}$.
 \label{lem:polypromoteR}
 \end{lem}
 \begin{proof} \mbox{} \\
 $\axctx \vdash \istp{A}$\lcase{by assumption}
 $\alphctx \vdash \istp{A}$\lcase{by $\copda {str} {alt}$}
 $\crel{\axctx}{\pctx}$\lcase{by assumption}
 $\crel{\alphctx}{\atpeqctx}$\lcase{by relational
   strengthening (Lemma~\ref{lem:strenR})}
 $\atpeqctx \vdash \tpaeq{A}{A}$\lcase{by Theorem~\ref{thm:prelreflR_tp}}
 $\pctx \vdash \tpaeq{A}{A}$\lcase{by $\copda {wk} {alt}$}
 \end{proof}
\begin{theorem}[Admissibility of Reflexivity for Terms, R Version]\\
  \label{thm:prelreflR}
  Let $\crel{\axctx}{\pctx}$.  If $\axctx \vdash \istm{M}$ then $\pctx
  \vdash \termeq{M}{M}$.
\end{theorem}
\begin{proof}
Again, the proof is by induction on the given derivation.  Most cases
are similar to the analogous cases in the proof of the R version for
the monomorphic case (Theorem~\ref{thm:reflR}) and the G version for
types in the polymorphic case (Theorem~\ref{thm:atpreflG}).
We show again the case for  application of terms to types to compare
with the G version.

 \noindent
Case:
$$\cD = \ibnc{\aboved{\cD_1}{\axctx \vdash\istm{M}}}
                   {\aboved{\cD_2}{\axctx \vdash\istp{A}}}
                   {\axctx \vdash \istm{(\tapp{M}{A})}}{}$$
$\crel{\axctx}{\pctx}$\lcase{by assumption}
$\axctx \vdash\istm{M}$\lcase{sub-derivation $\cD_1$}
$\pctx \vdash\termeq {M} M$\lcase {by IH}
$\axctx \vdash\istp {A}$\lcase {sub-derivation $\cD_2$}
$\pctx \vdash\istp {A}$\lcase {by Lemma~\ref{lem:polypromoteR}}
$\pctx \vdash \termeq{(\tapp{M}{A})}{(\tapp{M}{A})}$\lcase{by
  rule $ae_{ta}$}
\end{proof}

\subsection{Non-Linear Context Extensions}
\label{ssec:sub}

We return to the untyped lambda-calculus of
Sect.~\ref{ssec:basic} and establish the equivalence between the
algorithmic definition of equality defined previously, and declarative
equality $\dctx \vdash \termequal{M}{N}$, which includes reflexivity,
symmetry and transitivity in addition to the congruence
rules.\footnote{We acknowledge that this definition of declarative equality
has a degree of redundancy: the assumption $\termequal x x$ in
rule $de_l$ is not needed, since rule $de_r$ plays the variable
role. However, it yields an interesting generalized context schema, which exhibits
issues that would otherwise require more complex case studies.}
\[
\begin{array}{c}
\multicolumn{1}{l}{\mbox{Declarative Equality}}\\[0.75em]
\infer[de_v]{\Gamma \vdash \termequal{x}{x}}{\termequal{x}{x} \in \Gamma}
\qquad 
\infer[de_l]{\Gamma \vdash \termequal{(\lam{x}{M})}{(\lam{x}{N})}}{
\Gamma, \istm{x}; \termequal{x}{x} \vdash \termequal{M}{N}} 
\\[0.75em]
\infer[de_a]{\Gamma \vdash \termequal{(\app{M_1}{M_2})}{(\app{N_1}{N_2})}}
{\Gamma \vdash \termequal{M_1}{N_1} & \Gamma \vdash \termequal{M_2}{N_2}
}
\\[0.75em]
\infer[de_r]{\Gamma \vdash \termequal{M}{M}}{}
\quad
\infer[de_t]{\Gamma \vdash \termequal{M}{N}}{
\Gamma \vdash \termequal{M}{L} & \Gamma \vdash \termequal{L}{N}}
\quad
\infer[de_s]{\Gamma \vdash \termequal{M}{N}}{
\Gamma \vdash \termequal{N}{M}
}
\\[0.75em]
\multicolumn{1}{l}{\mbox{Context Schema }
\dsch ::= \istm x;\termequal{x}{x}}
\end{array}
\]

We now investigate the interesting part of the equivalence, namely that when
we have a proof of $(\termequal{M}{N})$ then we also have a proof of
$(\termeq{M}{N})$.  We show the G version first.

\subsubsection{G Version}
\label{subsubsec:ceqG}
Here, a generalized context must combine the atoms of $\actx$ and
$\dctx$ into one declaration:
\[
\begin{array}{llcl}
 \mbox{Generalized Context Schema} & S_{{da}} := 
 \istm x; \termequal x x; \termeq{x}{x}
\end{array}
\label{page:daschema}
\]
The following lemma promotes Theorems~\ref{thm:reflG}
and~\ref{le:sym-trans} to the ``bigger'' generalized context.
\begin{lem}[G-Promotion for Reflexivity, Symmetry, and Transitivity]
\label{lem:subsumpRST}
\begin{enumerate}
\item If $\Phi_{da} \vdash \istm{M}$, then
$\Phi_{da} \vdash \termeq{M}{M}$.

\item If $\Phi_{da} \vdash \termeq{M}{N}$, then
$\Phi_{da} \vdash \termeq{N}{M}$.

\item 
If $\Phi_{da} \vdash \termeq{M}{L}$ and
$\Phi_{da} \vdash \termeq{L}{N}$, then
$\Phi_{da} \vdash \termeq{M}{N}$.
\end{enumerate}
\end{lem}
\begin{proof}
Similar to the proof of Theorem~\ref{lem:polypromoteG} where the
application of $\copda {str} {decl}$ transforms a context
$\Phi_{da}$ to $\actx$ by considering each block of the form
$(\istm{x};\termequal{x}{x};\termeq{x}{x})$ and removing
$(\termequal{x}{x})$.
\end{proof}

\begin{theorem}[Completeness, G Version]
\label{thm:ceqG}  \mbox{}\\
If $\Phi_{da} \vdash \termequal{M}{N}$ then
  $\Phi_{da} \vdash \termeq{M}{N}$.
\end{theorem}

\begin{proof}
By induction on the derivation $\cD :: \Phi_{da} \vdash
\termequal{M}{N}$. We only show some cases.
\noindent

\noindent
Case: 
$$\cD = \ianc{}
             {\Phi_{da}\vd\termequal{M}{M}}{de_r}$$
$\Phi_{da} \vd\istm M$\lcase{by (implicit) assumption}
$\Phi_{da} \vdash \termeq{M}{M}$\lcase{by Lemma~\ref{lem:subsumpRST} (1)}

\bigskip\noindent
Case: $$\cD = \ibnc{\aboved{\cD_1}{\Phi_{da}\vdash \termequal{M}{L}}}
                   {\aboved{\cD_2}{\Phi_{da}\vdash \termequal{L}{N}}}
                   {\Phi_{da} \vdash \termequal{M}{N}}{de_t}$$
$\Phi_{da}\vdash \termeq{M}{L}$ and $\Phi_{da}\vdash
\termeq{L}{N}$\lcase{by IH on $\cD_1$ and $\cD_2$}
$\Phi_{da} \vdash \termeq{M}{N}$\lcase{by Lemma~\ref{lem:subsumpRST}
(3)}

\noindent
Case: 
$$\cD = \ianc{\aboved {\cD'}{\Phi_{da}, \istm x; \termequal{x}{x}
  \vdash \termequal{M}{N}}}{\Phi_{da}\vdash
  \termequal{(\lam{x}{M})}{(\lam{x}{N})}}{de_l}$$
$\Phi_{da},\istm{x};\termequal{x}{x};\termeq{x}{x}
 \vdash\termequal{M}{N}$ \hfill by $\dop wk$ on $\cD'$ \\
$\Phi_{da},\istm{x}; \termequal{x}{x}; \termeq{x}{x}
 \vdash \termeq{M}{N}$ \hfill
by IH \\
$\Phi_{da}, \istm{x};\termeq{x}{x} \vdash \termeq{M}{N}$ \hfill
by $\dop str$ \\
$\Phi_{da} \vdash \termeq{(\lam{x}{M})}{(\lam{x}{N})}$ \hfill
by rule $ae_l$
\end{proof}
The symmetry case is not shown, but also requires promotion, via
Lemma~\ref{lem:subsumpRST} (2). Note that the $de_l$ case requires
both $\dop{str}$ and $\dop{wk}$.  In contrast, the binder cases for
the G versions of the previous examples
(Theorems~\ref{thm:reflG},~\ref{thm:atpreflG},
and~\ref{thm:prelreflG}) required only $\dop{wk}$.  The need for both
arises from the fact that the generalized context is a non-linear
extension of two contexts, i.e., it is not the same as either
one of the two contexts it combines.

\subsubsection{R Version}
\label{sssec_completeR}
The context relation required here is $\crel \actx \dctx$:
\[
\begin{array}{c}
  \multicolumn{1}{l}{\mbox{Context Relation}}\\[0.75em]
\infer[crel_{ad}]{\crel {\actx, \istm{x};\termeq x x}
                 {\dctx,  \istm x; \termequal x x}}{\crel \actx \dctx} 
\end{array}
\]
As in Sect.~\ref{ssec:poly}, we need the appropriate promotion
lemma, which again requires a relation strengthening lemma:
\begin{lem}[Relational Strengthening]
Let $\crel{\actx}{\dctx}$.  Then there exists a context 
$\ictx$ such that $\crel{\ictx}{\actx}$.
\label{lem:completestrenR}
\end{lem}
 \begin{lem}[R-Promotion for Reflexivity]
   Let $\crel{\actx}{\dctx}$. If $\dctx \vdash \istm{M}$ then $\actx
   \vdash \termeq{M}{M}$.
 \label{lem:completepromoteR}
 \end{lem}
%
The proofs are analogous to Lemmas~\ref{lem:strenR}
and~\ref{lem:polypromoteR},
with the proof of Lemma~\ref{lem:completepromoteR} requiring
Lemma~\ref{lem:completestrenR}.

\begin{theorem}[Completeness, R Version]
\label{thm:ceqR}
Let $\crel\actx\dctx$. 
If $\dctx \vdash \termequal{M}{N}$ then $\actx \vdash \termeq{M}{N}$.
\end{theorem}

\begin{proof} By induction on the derivation $\cD :: \dctx
\vdash \termequal{M}{N}$. 

\noindent
Case:
$$\cD = \ianc{}
             {\dctx\vd\termequal{M}{M}}{de_r}$$
$\dctx\vd\istm M$\lcase{by (implicit) assumption}
$\actx \vdash \termeq{M}{M}$\lcase{by Theorem~\ref{lem:completepromoteR}}


\bigskip\noindent
Case: $$\cD = \ibnc{\aboved{\cD_1}{\dctx\vdash \termequal{M}{L}}}
                   {\aboved{\cD_2}{\dctx\vdash \termequal{L}{N}}}
                   {\dctx \vdash \termequal{M}{N}}{de_t}$$
$\actx\vdash \termeq{M}{L}$ and $\actx\vdash
\termeq{L}{N}$\lcase{by IH on  $\cD_1$ and $\cD_2$}
$\actx \vdash \termeq{M}{N}$\lcase{by Theorem~\ref{le:sym-trans} (2)}

\noindent
Case: 
$$\cD = \ianc{\aboved {\cD'}{\dctx, \istm x; \termequal{x}{x}
  \vdash \termequal{M}{N}}}{\dctx\vdash
  \termequal{(\lam{x}{M})}{(\lam{x}{N})}}{de_l}$$
%
$\crel{\actx}{\dctx}$ \hfill by assumption \\
$\crel {\actx,\istm{x};\termeq{x}{x}} {\dctx,\istm{x};\termequal{x}{x}}$
\hfill by rule $crel_{ad}$ \\
$\actx,\istm{x};\termeq{x}{x}
 \vdash\termeq{M}{N}$ \hfill by IH on $\cD'$  \\
$\actx \vdash \termeq{(\lam{x}{M})}{(\lam{x}{N})}$ \hfill
by rule $ae_l$
\end{proof}
Only one promotion lemma is required in this proof, for
the reflexivity case (which requires one occurrence each of
$\copda{str}{}$ and $\copda{wk}{}$), and no strengthening or weakening is
needed in the lambda case (thus no occurrences of $\dop{str/wk}$ in
this proof). 
In contrast, the proof of the G version of this theorem
(Theorem~\ref{thm:ceqG}) uses 3 occurrences of each of $\copda{str}{}$
and $\copda{wk}{}$ via promotion Lemma~\ref{lem:subsumpRST} and one
occurrence each of $\dop{str}$ and $\dop{wk}$ in the lambda case.

\subsection{Order
}
\label{ssec:exc}

A consequence of viewing contexts as {sequences} is that \emph{order}
comes into play, and therefore the need to consider \emph{exchanging}
the elements of a context. This happens when, for example, a judgment
singles out a particular occurrence of an assumption in head
position. We exemplify this with a ``parallel'' substitution property for
algorithmic equality, stated below.
The proof also involves some slightly more sophisticated reasoning
about names in the variable case than previously observed. Furthermore,
note that this substitution property does not ``come for free'' in a HOAS
encoding in the way, for example, that type substitution
(Lemma~\ref{lem:of-subst}) does.
\begin{theorem}[Pairwise Substitution]
\label{thm:asubst}
  If $\actx, \istm x; \termeq x x\vd\termeq {M_1} {M_2} $ and $\actx\vd\termeq
  {N_1} {N_2}$, then $\actx\vd\termeq {([N_1/x]M_1)} {([N_2/x]M_2)}$.
\end{theorem}
\begin{proof}
  By induction on the derivation $\cD :: \actx, \istm x;\termeq
  x x\vd\termeq {M_1} {M_2}$ and inversion on $\actx\vd\termeq {N_1}
  {N_2}$.  We show two cases.

\noindent
Case:
$$\cD = \ianc{\termeq {y}{y}  \in \actx, \istm x; \termeq x x}
             {\actx, \istm x;\termeq x x \vdash \termeq {y}{y} }{ae_v}$$
We need to establish $\actx\vd\termeq {([N_1/x]y)} {([N_2/x]y)}$.

\noindent
\emph{Sub-case:} $y = x$: Applying the substitution to the above
judgment, we need to show $\actx\vd\termeq{N_1}{N_2}$,
which we have.

\noindent
\emph{Sub-case:} $\termeq {y}{y}  \in \actx$, for  $y\not =
x$. Applying the substitution in this case gives us
$\actx\vd\termeq y y$, which we have by assumption. 

\noindent
Case:
$$\cD = \ianc{\aboved {\cD'}{\actx , \istm x; \termeq x x, \istm y;
    \termeq y y
  \vdash \termeq{M_1}{M_2}}}{\actx, \istm x; \termeq x x \vdash
  \termeq{(\lam{y}{M_1})}{(\lam{y}{M_2})}}{de_l}$$
%
$\actx ,  \istm y; \termeq y y, \istm x;\termeq x x 
  \vdash \termeq{M_1}{M_2}$ \hfill by $exc$ on  $\cD'$\\
 $\actx\vd\termeq  {N_1} {N_2}$,\hfill by assumption\\
$\actx,\istm y;\termeq{y}{y}\vd\termeq  {N_1} {N_2}$\hfill by $\dop wk$\\
$\actx,\istm y;\termeq{y}{y}\vd\termeq {([N_1/x]M_1)} {([N_2/x]M_2})$ \hfill
by IH \\
$\actx  \vdash \termeq{[N_1/x](\lam{y}{M_1})}{[N_2/x](\lam{y}{M_2})}$ \hfill
by rule $ae_l$ and possible renaming.
\end{proof}

We remark that there are more general ways to formulate properties such as
Theorem~\ref{thm:asubst} that do \emph{not} require (on paper) exchange; for
example, 
\begin{quote}
  If $\actx,\istm x; \termeq x x, \actx'\vd\termeq {M_1} {M_2} $
  and $\actx\vd\termeq {N_1} {N_2}$, then $\actx, \actx'\vd\termeq
  {([N_1/x]M_1)} {([N_2/x]M_2)}$.
\end{quote}
The proof of the latter statement has a similar structure to the
previous one, except that it uses $\dop wk$ in the first variable
sub-case, while the binding case does not employ any structural
property to apply the induction hypothesis, by taking $(\actx'',\istm
y;\termeq{y}{y})$ as $\actx'$. While this works well in a paper and
pencil style, it is much harder to mechanize, since it brings in
reasoning about appending and splitting lists that are foreign to the
matter at hand.

We conclude by noting that there are examples where exchange cannot be
applied, since the dependency proviso is not satisfied. Cases in point
are substitution lemmas for dependent types. Here, other encoding
techniques must be used, as explored in  \citet{Crary09}.

\subsection{Uniqueness
}
\label{ssec:unique}
\renewcommand{\lamt}[3]{\mathsf{lam}\,#1^{#2}.\,#3}
Uniqueness of context variables plays an unsurprisingly important role in
proving type uniqueness, i.e.\ every lambda-term has a unique
type. For the sake of this discussion it is enough to consider the
monomorphic case, where abstractions include type annotations on bound
 variables, and types consist only of a ground type and a function
 arrow.
 \[
 \begin{array}[t]{llll}
 \mbox{Terms} & M & \bnfas & y \mid \lamt{x}{A}{M} \mid \app{M_1}{M_2}\\
 \mbox{Types} & A & \bnfas & \grnd \mid \fsp A B
 \end{array}
 \] 
The typing rules are the obvious subset
 of the ones presented in Sect.~\ref{sec:theory}, yielding:
\[
\begin{array}{llcl}
 \mbox{Context Schema} & \tsch:= \istm x; x \oftp A
\end{array}
\]
The statement of the theorem requires only a single context and thus
there is no distinction to be made between the R and G versions.

\begin{theorem}[Type Uniqueness]\label{thm:unique} 
 If $\tctx \vdash M \hastype A$ and $\tctx \vdash M \hastype B$ then $A = B$.
\end{theorem}
\begin{proof}
  The proof is by induction on the first derivation and inversion on
  the second.  We show only the variable case where uniqueness plays a
  central role.

\noindent
Case: $$\cD = \ianc{x\oftp A\in \tctx}{\gvd x \hastype A}{{\mathit{of}_v}}$$
We know that 
$x \oftp A \in \tctx$ by rule $\mathit{of}_v$. By definition, $\tctx$
contains block $(\istm{x};x\oftp A)$.  Moreover, we know $\tctx \vdash
x \hastype B$ by assumption. By inversion using rule $\mathit{of}_v$,
we know that $x \oftp B \in \tctx$, which means that $\tctx$ contains
block $(\istm{x};x\oftp B)$.  Since all assumptions about $x$ occur
uniquely, these must be the same block. Thus $A$ must be identical to
$B$. 
\end{proof}

\subsection{Substitution }
\label{ssec:cut}

In this section we address the interaction of the substitution
property with context reasoning. It is well known and rightly
advertised that substitution lemmas come ``for free'' in HOAS
encodings, since substitutivity is just a by-product of
hypothetical-parametric judgments. We refer
to~\citet{Pfenning01book} for more details.  A classic
example is the proof of type preservation for a functional programming
language, where a lemma stating that substitution preserves typing is
required in every case that involves a $\beta$-reduction.  
However, this example theorem is unduly restrictive since
functional programs are closed expressions; in fact, the proof
proceeds by induction on (closed) evaluation and inversion on typing,
hence only addressing contexts in a marginal way. We thus discuss a
similar proof for an evaluation relation that ``goes under a lambda''
and we choose parallel reduction, as it is a standard relation also
used in other important case studies such as the Church-Rosser
theorem. The context schema and relevant rules are below.
\[  \begin{array}{c}
 \multicolumn{1}{l}{\mbox{Parallel Reduction}}\\[0.75em]
 \ianc{x\step x\in \Gamma}{\Gamma\vd x \step x}{{pr_v}}\qquad 
\ianc{\Gamma, \istm{x};x\step x \vd M \step N}
     {\Gamma\vd \lam x M \step \lam x N }{pr_l} \vsk
 \ibnc{\Gamma, \istm{x};x\step x \vd M \step M'}{\Gamma\vd N \step N'}
{\Gamma\vd(\app{(\lam x M)}  N) \step [N'/x] M' }{{pr_{\beta}}} \vsk
 \ibnc{\Gamma\vd M \step M'}{\Gamma\vd N\step N'}{\Gamma\vd(\app M  N) \step (\app {M'}{N'})}{{pr_a}}\qquad
\\[0.75em]
\multicolumn{1}{l}{\mbox{Context Schema }
\rsch := \istm x; x \step x}
\end{array}
\]

The relevant substitution lemma  is:
\begin{lem}
\label{lem:of-subst}
  If $\tctx, \istm{x};x\oftp A\vd M\hastype B$ and $\gvd N\hastype A$, then
  $\gvd [N/x]M\hastype B$. 
\end{lem}
\begin{proof} While this is usually proved by induction on the first
  derivation, we show it as a corollary of the substitution
  principles.

\noindent
 $\tctx, \istm{x};x\oftp A\vd M\hastype B$ \hfill by assumption\\ 
 $\tctx, \istm{N};N\oftp A\vd [N/x]M\hastype B$ \hfill by parametric
 substitution\\ 
 $\tctx, \istm{N}\vd [N/x]M\hastype B$ \hfill by hypothetical
 substitution\\ 
$\tctx\vd \istm N$\hfill by (implicit) assumption\\
  $\gvd [N/x]M\hastype B$\hfill by hypothetical
 substitution
\end{proof}

We show only the R version of type preservation.  For the G version,
the context schema is obtained by combining the schemas $\rsch$
and $\tsch$ similarly to how $S_{{da}}$ was defined to
combine $\asch$ and $\dsch$ in
Sect.~\ref{subsubsec:ceqG}.  We leave it to the reader to
complete such a proof.  For the R version, we introduce the customary
context relation, which in this case is:
\[
\begin{array}{c}
  \multicolumn{1}{l}{\mbox{}}
\infer[crel_{rt}]{\crel{\rctx, \istm{x};x\step x}
               {\tctx, \istm{x};x\oftp A}}{\crel \rctx \tctx}
\end{array}
\]

\begin{theorem}[Type Preservation for Parallel Reduction]
\label{thm:tps}
Assume $\crel\rctx\tctx$. If $\pvd M\step N$ and $\gvd M\hastype A$,
then $\gvd N\hastype A$.
\end{theorem}
\begin{proof}
  The proof is by induction on the derivation $\cD :: \pvd M\step N$ and
  inversion on $\gvd M\hastype A$. We show only two cases:

\medskip\noindent
Case: $$\cD =  \ianc{x\step x\in \rctx}{\pvd x \step x}{pr_v}$$
We know that in this case $M=x=N$. Then the result follows trivially.

\medskip\noindent
Case: $$\cD = \ibnc{\aboved{\cD_1}{\rctx, \istm{x};x\step x \vd M \step M'}}
                   {\aboved{\cD_2}{\pvd N \step N'}}
                   {\pvd(\app{(\lam x M)}  N) \step [N'/x] M'}{pr_{\beta}}$$
$\gvd (\app{(\lam x M)} N) \hastype A$\hfill by assumption\\
${\gvd (\lam x M) \hastype \fsp{B}{A}}$ and ${\gvd N \hastype B}$\lcase{by
  inversion on rule $\mathit{of}_a$}
${\gvd N' \hastype B}$\lcase{by IH on $\cD_2$ and the latter}
$\tctx, \istm{x};x\oftp B \vd M \hastype A$ \lcase{by
  inversion on rule $\mathit{of}_l$}
$\crel\rctx\tctx $ \lcase{by assumption}
$\crel{(\rctx, \istm{x};x\step x)}{(\tctx, \istm{x};x \oftp B)}$
   \lcase{by rule $crel_{rt}$}
$\tctx, \istm{x};x\oftp B \vd M' \hastype A$ \lcase{by IH}
 $\gvd [N'/x]M'\hastype A$ \lcase{by Lemma~\ref{lem:of-subst} (substitution)}
\end{proof}


If we were to prove a similar result for the polymorphic
$\lambda$-calculus, we would need another substitution lemma, namely:
\begin{lem}
\label{lem:poly-subst}
\sloppypar
  If $\atpctx, \istp{\alpha}\vd M\hastype B$ and $\atpctx \vd \istp A$, then
 \mbox{ $\atpctx \vd [A/\alpha]M\hastype [A/\alpha]B$}. 
\end{lem}
Again, this follows immediately from parametric and hypothetical
substitution, whereas a direct inductive proof may not be completely
trivial to mechanize.
\section{The ORBI Specification Language}
\label{sec:mech}

{ORBI} (\underline{O}pen challenge problem
\underline{R}epository for systems supporting reasoning with
\underline{BI}nders) is an open repository for sharing benchmark
problems based on the notation we have developed.
ORBI is designed to be a human-readable, easily machine-parsable,
uniform, yet flexible and extensible language for writing
specifications of formal systems including grammar, inference rules,
contexts and theorems. The language directly upholds HOAS
representations and is oriented to support the mechanization of the
benchmark problems in Twelf, Beluga, Abella, and Hybrid, without
hopefully precluding other existing or future HOAS systems. At the
same time, we hope it also is amenable to translations to systems
using other representation techniques such as nominal systems.

The desire for ORBI to cater to both type and proof theoretic
frameworks requires an almost impossible balancing act between the two
views. While all the systems we plan to target are essentially
two-level, they differ substantially, as we will see in much more
detail in the companion paper~\cite{companion}. For example, contexts
are first-class and part of the specification language in Beluga; in
Twelf, schemas for contexts are part of the specification language,
which is an extension of LF, but users cannot explicitly quantify over
contexts and manipulate them as first-class objects; in Abella and
Hybrid, contexts are (pre)defined using inductive definitions on the
reasoning level.

We structure the language in two parts:
 \begin{enumerate}
 \item the problem description, which includes the grammar of
   the object language syntax, inference rules, context schemas and
   context relations;
 \item the logic language, which includes syntax for expressing
   theorems and directives to ORBI2X\footnote{Following TPTP's
     nomenclature \cite{TPTP}, we call ``ORBI2X'' any tool taking an
     ORBI specification as input; for example, the translator for
     Hybrid mentioned earlier translates syntax, inference rules, and
     context definitions of ORBI into input to the Coq version of
     Hybrid, and is designed so that it can be adapted fairly directly
     to output Abella scripts.}  tools.
 \end{enumerate}

 We consider the notation that we present here as a first attempt at
 defining ORBI (Version 0.1), where the goal is to cover the
 benchmarks considered in this paper.  As new benchmarks are added, we
 are well aware that we will need to improve the syntax and increase
 the expressive power---we discuss limitations and some possible
 extensions in Sect.~\ref{sec:concl}.

\subsection{Problem Description}
ORBI's language for defining the grammar of an object language
together with inference rules is based on the logical framework LF;
pragmatically, we have adopted the concrete syntax of LF
specifications in Beluga which is almost identical to  Twelf's.
The advantage is that specifications can be directly type checked by
Beluga thereby
eliminating many syntactically correct but meaningless expressions.

Object languages are written according to the EBNF (Extended
Backus-Naur Form) grammar in Fig.~\ref{fig:gramlf}, which uses certain
conventions: \verb!{a}!  means repeat a production zero or more times,
and comments in the grammar are enclosed between \verb!(*! and
\verb!*)!.  The token \verb!id! 
refers to identifiers starting with a lower or upper case letter.
\begin{figure}[th]
\begin{verbatim}
sig      ::= {decl                      (* declaration *)
           | s_decl}                    (* schema declaration *)

decl     ::= id ":" tp "."              (* constant declaration *)
           | id ":" kind "."            (* type declaration *)

op_arrow ::= "->" | "<-"                (* A <- B same as B -> A *)

kind     ::= type
           | tp op_arrow kind           (* A -> K *)
           | "{" id ":" tp "}" kind     (* Pi x:A.K *)

tp       ::= id {term}                  (* a M1 ... M2 *)
           | tp op_arrow tp
           | "{" id ":" tp "}" tp       (* Pi x:A.B *)

term     ::= id                         (* constants, variables *)
           | "\" id "." term            (* lambda x. M *)
           | term term                  (* M N *)

s_decl    ::= schema s_id ":" alt_blk "."

s_id      ::= id

alt_blk   ::= blk {"+" blk}                         

blk       ::= block id ":" tp {";" id ":" tp}  
\end{verbatim}
\caption{ORBI Grammar for Syntax, Judgments, Inference Rules, and
  Context Schemas}
\label{fig:gramlf} 
\end{figure}
These grammar rules are basically the standard ones used both in Twelf
and Beluga and we do not discuss them in detail here. We only note
that while the presented grammar permits general dependent types up to
level $n$, ORBI specifications will only use level 0 and level
1. Intuitively, specifications at level 0 define the syntax of a given
object language, while specifications at level 1 (i.e.~type families
which are indexed by terms of level 0) describe the judgments and
rules for a given OL\@.  We exemplify the grammar relative to the
example of algorithmic vs.\ declarative equality used in
Subsections~\ref{ssec:basic},~\ref{ssec:sub}, and~\ref{ssec:exc}. The
full ORBI specification is given in Appendix~\ref{sec:aede-orbi}, and
all examples described in this section are taken from that
specification.  For the remaining example specifications, we refer the
reader to the the companion paper \cite{companion} or to
\url{https://github.com/pientka/ORBI}.

To assist compact translations to systems that do not include the LF
language, we also support \emph{directives} written as comments of a
special form, i.e., they are prefixed by \verb!%! and
ignored by the LF type checker.  For example, we provide directives
that allow us to distinguish between the syntax definition of an
object language and the definition of its judgments and inference
rules. (See Appendix~\ref{sec:aede-orbi}.)
Directives, including their grammar, are detailed in
Sect.~\ref{ssec:direct}.

\paragraph{Syntax} 
An ORBI file starts in the \verb!Syntax! section with the declaration
of the constants used to encode the syntax of the OL in question, here
untyped lambda-terms, which are introduced with the declaration
\verb!tm:type!.  This declaration along with those of the constructors
\verb!app! and \verb!lam! in the \verb!Syntax! section fully specify
the syntax of OL terms.  We represent binders in the OL
 using binders in the HOAS
meta-language. Hence the constructor \verb!lam! takes in a
function of type \verb!tm -> tm!. For example, the
OL term $(\lam{x}{\lam{y}{\app{x}{y}}})$ is represented as
\verb!lam (\x. lam (\y. app x y))!, where ``$\verb-\-$'' is the binder
of the metalanguage. Bound variables found in the object language are
not explicitly represented in the meta-language.

\paragraph{Judgments and Rules} These are introduced as LF type
families (predicates)
in the \verb|Judgments| section followed by object-level inference
rules for these judgments in the \verb|Rules| section.\footnote{There
  are several excellent
  tutorials~\cite{Pfenning01book,HarperLicata:JFP07} on how to encode
  OLs in LF, and hence we keep it brief.}  In our running example, we
have two judgments, \verb|aeq| and \verb|deq| of type
\verb|tm -> tm -> type|.  Consider first the inference rule for
algorithmic equality for application, where the ORBI text is a
straightforward encoding of the rule:

\noindent
\begin{minipage}[c]{0.5\linewidth}
\begin{small}
\begin{verbatim}
ae_a: aeq M1 N1 -> aeq M2 N2 
   -> aeq (app M1 M2) (app N1 N2).
\end{verbatim}   
\end{small}
\end{minipage}
\begin{minipage}[c]{0.4\linewidth}
\[
\begin{array}{c}
~~~~~\infer[ae_a]{\termeq{(\app{M_1}{M_2})}(\app{N_1}{N_2})}
 {\termeq{M_1} {N_1}\qquad  \termeq{M_2}{N_2}
 }
\\\;
\end{array}
\]
\end{minipage}

\noindent
Uppercase letters such as \verb|M1| denote schematic variables,
which are implicitly quantified at the outermost level, namely
\verb|{M1:tm}|, as commonly done for readability purposes in Twelf and
Beluga.

The binder case is  more interesting:

\noindent
\begin{minipage}[c]{0.6\linewidth}
\begin{small}
\begin{verbatim}
ae_l: ({x:tm} aeq x x -> aeq (M x) (N x)) 
   -> aeq (lam (\x. M x)) (lam (\x. N x)).
\end{verbatim}
\end{small}
\end{minipage}
\begin{minipage}[c]{0.1\linewidth}
\[
\begin{array}{c}
\infer[ae_l^{x,ae_v}]{\termeq{(\lam{x}{M})}{(\lam{x}{N})}}{
 \begin{array}{c}
  \infer[x]{\istm{x}}{}\hspace{1cm}\infer[ae_v]{\termeq{x}{x}}{}\\
\vdots\\
\termeq{M}{N}
 \end{array}
} 
\end{array}
\]
\end{minipage}

\noindent
We 
view the $\istm x$ assumption as the parametric
assumption \lstinline{x:tm}, while the hypothesis $\termeq x x$
(and its scoping) is encoded within the embedded implication
\lstinline{aeq x x -> aeq (M x) (N x)} in the current (informal)
signature augmented with the dynamic declaration for
\lstinline{x}.\footnote{As is well known, 
parametric assumptions and embedded implication are
unified in the type-theoretic view.} Recall that the ``variable'' case of
an implicit-context presentation, namely $ae_v$, is folded inside
the binder case.

\paragraph{Schemas} A {schema}
declaration \verb!s_decl! is introduced
using the keyword \verb!schema!.  A \verb!blk! consists of one or more
declarations and \verb!alt_blk! describes \emph{alternating}
schemas. For example, schema  $\asch$ in Sect.~\ref{sssec:gen}
appears in the \verb!Schemas! section of Appendix~\ref{sec:aede-orbi} as:
\begin{verbatim}
schema xaG: block (x:tm; u:aeq x x).
\end{verbatim}  
As another example, in this case illustrating
a schema sporting alternatives, we encode the
schema $\psch$ from polymorphic 
equality as:

\begin{verbatim}
schema aeqG: block (a:tp; u:atp a a) + block (x:tm; v:aeq x x).
\end{verbatim}  

While we can type-check the schema definitions using an extension of
the LF type checker (as implemented in Beluga), we
do not verify that the given schema definition is meaningful with
respect to the specification of the syntax and inference rules; in
other words, we do not perform ``world checking'' in Twelf lingo.

\paragraph{Definitions}
So far we have considered the specification language for encoding
formal systems. ORBI also supports declaring inductive definitions for
specifying context relations and theorems.  We start with the grammar
for inductive definition (Fig.~\ref{fig:gramindef}). An inductive
predicate is given a \verb|r_kind| by the production \verb!def_dec!.
Although we plan to provide syntax for specifying more general
inductive definitions, in this version of ORBI we \emph{only} define
\emph{context relations} inductively, that is $n$-ary predicates
between contexts of a given schema. Hence the base predicate is of the
form \verb|id {ctx}| relating different contexts.

\begin{figure}[ht]
\begin{verbatim}
def_dec  ::= "inductive" id ":" r_kind "=" def_body "."

r_kind   ::= "prop"
           | "{" id ":" s_id "}" r_kind

def_body ::= "|" id ":" def_prp {def_body}

def_prp  ::= id {ctx}
           | def_prp "->" def_prp

ctx      ::= nil | id | ctx "," blk
\end{verbatim}
\caption{ORBI Grammar for Inductive Definitions describing Context Relations}
\label{fig:gramindef} 
\end{figure}
\noindent
For example, the relation $\crel \ictx \actx$ is encoded in the
\verb!Definitions! section of Appendix~\ref{sec:aede-orbi} as:
\begin{verbatim}
inductive xaR : {G:xG} {H:xaG} prop =
| xa_nil: xaR nil nil
| xa_cons: xaR G H -> xaR (G, block x:tm) (H, block x:tm; u:aeq x x).
\end{verbatim}
This kind of  relation can be translated fairly
directly to inductive n-ary predicates in systems
supporting the proof-theoretic view. In the type-theoretic framework
underlying Beluga, inductive predicates relating contexts correspond
to recursive data types indexed by contexts; this also allows for a
straightforward translation. Twelf's type theoretic framework,
however, is not rich enough to support inductive definitions.

\subsection{Language for Theorems and Directives}

While the elements of an ORBI specification detailed in the previous
subsection were relatively easy to define in a manner that is well
understood by all the different systems we are targeting, we
illustrate in this subsection those elements that are harder to
describe uniformly due to the different treatment and meaning of
contexts in the different systems.

\paragraph{Theorems}
We list the grammar for theorems in Fig.~\ref{fig:gramthm}.  Our
reasoning language includes a category \verb!prp! that specifies the logical
formulas we support. The base predicates include \verb!false,true!,
term equality, atomic predicates of the form \verb!id {ctx}!, which
are used to express context relations, and predicates of the form
\verb![ctx |- J]!, which represent judgments of an object language
within a given context.
Connectives and quantifiers include implication, conjunction,
disjunction, universal and existential quantification over terms, and
universal quantification over context variables.

\begin{figure}[ht]
\begin{verbatim}
thm      ::= "theorem" id ":" prp "."

prp      ::= id {ctx}                     (* Context relation *)
           | "[" ctx  "|-" id {term} "]"  (* Judgment in a context *)
           | term "=" term                (* Term equality *)
           | false                        (* Falsehood *)
           | true                         (* Truth *)
           | prp "&" prp                  (* Conjunction *)
           | prp "||" prp                 (* Disjunction *)
           | prp "->" prp                 (* Implication *)
           | quantif prp                  (* Quantification *)

quantif  ::= "{" id ":" s_id "}"          (* universal over contexts *)
           | "{" id ":" tp "}"            (* universal over terms *)
           | "<" id ":" tp ">"            (* existential over terms *)
\end{verbatim}
\caption{ORBI Grammar for Theorems}
\label{fig:gramthm} 
\end{figure}

The specification of the G and R versions
of the completeness theorem is as follows:
\begin{verbatim}
theorem ceqG: {G:daG} [G |- deq M N] -> [G |- aeq M N].
theorem ceqR: {G:xdG}{H:xaG} daR G H -> [G |- deq M N] -> [H |- aeq M N].
\end{verbatim}
This and all the others theorems pertaining to the development of the
meta-theory of algorithmic and declarative equality are listed in the
\verb!Theorems! section of Appendix~\ref{sec:aede-orbi}. The theorems
stated are a straightforward encoding of the main theorems in 
Subsections~\ref{ssec:basic},~\ref{ssec:sub}, and~\ref{ssec:exc}.

As mentioned, we do not type-check theorems; in particular, we do not
define the meaning of \verb![ctx |- J]!, since several interpretations
are possible.  In Beluga, every judgment \verb+J+ must be meaningful
within the given context \verb+ctx+; in particular, \emph{terms}
occurring in the judgment \verb+J+ must be meaningful in
\verb+ctx+. As a consequence, both parametric and hypothetical
assumptions relevant for establishing the proof of \verb+J+ must be
contained in \verb+ctx+. Instead of the local context view adopted in
Beluga, Twelf has one global ambient context containing all relevant
parametric and hypothetical assumptions. Systems based on proof-theory
such as Hybrid and Abella distinguish between assumptions denoting
eigenvariables (i.e.~parametric assumptions), which live in a global
ambient context and proof assumptions (i.e.~hypthetical assumptions),
which live in the context \verb+ctx+. While users of different systems
understand how to interpret \verb![ctx |- J]!, reconciling these
different perspectives in ORBI is beyond the scope of this paper. Thus
for the time being, we view theorem statements in ORBI as a kind of
\emph{comment}, where it is up to the user of a particular system to
determine how to translate them.

\paragraph{Directives}
\label{ssec:direct}

As we have mentioned before, \emph{directives} are comments that help
the ORBI2X tools to generate target
representations of the ORBI specifications. The idea is reminiscent of
what  \emph{Ott}  \cite{ott} does to customize certain
declarations, e.g.\ the representation of variables, to the different
programming languages/proof assistants it supports. 
The grammar for directives is listed in Fig.~\ref{fig:direct}.

\begin{figure}[th]
  \centering
\begin{verbatim}
dir       ::=  '%' sy_id what decl {dest} '.'
             | '%%' sepr '.'

sy_id     ::=  hy | ab | bel | tw 

sy_set    ::= '[' sy_id {',' sy_id} ']'

what      ::=  wf | explicit | implicit

dest      ::= 'in' ctx | 'in' s_id | 'in' id

sepr      ::=  Syntax | Judgments | Rules | Schemas | Definitions 
               | Directives | Theorems
\end{verbatim}  
  \caption{ORBI Grammar for Directives}
  \label{fig:direct}
\end{figure}

Most of the directives that we consider in this version of ORBI are
dedicated to help the translations into proof-theoretical systems,
although we include also some to facilitate the translation of
\emph{theorems} to Beluga. The set of directives is not intended to be
complete and the meaning of directives is system-specific.  Beyond
directives (\texttt{sepr}) meant to structure ORBI specs, the
instructions \verb|wf| and \verb|explicit| are concerned with the
asymmetry in the proof-theoretic view between declarations that give
typing information, e.g. \verb|tm:type|, and those expressing
judgments, e.g. \verb|aeq:tm -> tm -> type|. In Abella and Hybrid, the
former may need to be reified in a judgment, in order to show that
judgments preserve the well-formedness of their constituents, as well
as to provide induction on the structure of terms; yet, in order to
keep proofs compact and modular, we want to minimize this reification
and only include them where necessary.  The first line in the
\verb!Directives! section of Appendix~\ref{sec:aede-orbi} states the
directive ``\verb|% [hy,ab] wf tm|'' that refers to the first line of
the
\verb!Syntax! section where \verb!tm! is introduced, and indicates
that we need a predicate (e.g., \verb|is_tm|) to express
well-formedness of terms of type \verb|tm|.  Formulas expressing the
definition of this predicate are automatically generated from the
declarations of the constructors \lstinline{app} and \lstinline{lam}
with their types.

The keyword \verb|explicit| indicates when such well-formedness
predicates should be included in the translation of the declarations
in the \verb!Rules! section.  For example, the following
formulas both represent possible translations of the \verb|ae_l| rule
to Abella and Hybrid:
$$\begin{array}{l}
\forall M,N.\;
(\forall x.\; \mathsf{is\_tm}~x\rightarrow\termeq{x}{x}\rightarrow
 \termeq{M x}{N x})\rightarrow
 \termeq{(\mathsf{lam}~M)}{(\mathsf{lam}~N)} \\
\forall M,N.\;
(\forall x.\; \termeq{x}{x}\rightarrow
 \termeq{M x}{N x})\rightarrow
 \termeq{(\mathsf{lam}~M)}{(\mathsf{lam}~N)}
\end{array}$$
where the typing information is explicit in the first and implicit in
the second.  By default, we choose the latter, that is well-formed
judgments are assumed to be \emph{implicit}, and require a directive
if the former is desired. 
In fact, in the previous section, we assumed
that whenever a judgment is provable, the terms in it are
well-formed, e.g., if \verb|aeq M N| is provable, then so are \verb|is_tm M|
and \verb|is_tm N|.  Such a lemma is indeed provable in Abella and
Hybrid from the \emph{implicit} translation of the rules for
\verb|aeq|.  Proving a similar lemma for the \verb|deq| judgment, on
the other hand, requires some strategically placed explicit
well-formedness information.  In particular, the two
directives 
\begin{verbatim}
% [hy,ab] explicit x in de_l.
% [hy,ab] explicit M in de_r.
\end{verbatim}
require the clauses \verb|de_l| and \verb|de_r| to be translated to the
following formulas:
$$\begin{array}{l}
\forall M,N.\;
(\forall x.\; \mathsf{is\_tm}~x\rightarrow\termequal{x}{x}\rightarrow
 \termequal{M x}{N x})\rightarrow
 \termequal{(\mathsf{lam}~M)}{(\mathsf{lam}~N)} \\
\forall M.\; \mathsf{is\_tm}~M\rightarrow\termeq{M}{M}
\end{array}$$

The case for schemas is analogous: in 
the proof-theoretic view, schemas are translated to 
unary inductive predicates.  Again, typing information is left
implicit in the translation unless a directive is included.  For
example, the \verb|xaG| schema with no associated directive will be
translated to a definition that expresses that whenever context
\verb|G| has schema \verb|xaG|, then so does \verb|G,aeq x x|.  For
the \verb|daG| schema, with directive

\begin{verbatim}
% [hy,ab] explicit x in daG.
\end{verbatim}

\begin{sloppypar}
\noindent the translation will express that whenever \verb|G| has schema
\verb|daG|, then so does \verb|G, (is_tm x;deq x x;aeq x x)|.
\end{sloppypar}

Similarly, directives in context relations, such as:
\begin{verbatim}
% [hy,ab] explicit x in G in xaR.
\end{verbatim}
also state which well-formedness annotations to make explicit in the
translated version. In this case, when translating the definition of
\verb!xaR! in the \verb!Definitions! section, they are to be kept in
\verb|G|, but skipped in \verb|H|.

Keeping in mind that we consider the notion of directive
\emph{open} to cover other benchmarks and different systems, we offer
some speculation about directives that we may need to translate
theorems for the examples and systems that we are considering.
(Speculative directives are omitted from
Appendix~\ref{sec:aede-orbi}).  For example, theorems \verb|reflG| is
proven by induction over \verb+M+. As a consequence, \verb+M+ must be
explicit.
\begin{verbatim}
% [hy,ab,bel] explicit M in H in reflG.
\end{verbatim}
\begin{sloppypar}
\noindent
Hybrid and Abella interpret the directive by adding an explicit assumption 
$\left[H \vdash \mathsf{is\_tm}~M\right]$, as illustrated by the result of the translation: 
\end{sloppypar}
$$\forall H,M.\; \left[H \vdash \mathsf{is\_tm}~M\right]\rightarrow
               \left[H \vdash \termeq{M}{M}\right]$$
In Beluga, the directive is interpreted as
\begin{lstlisting}
{H:xaG} {M:[H.tm]} [H.aeq (M ..) (M ..)].
\end{lstlisting}
where \lstinline!M! will have type \lstinline!tm! in the context
\lstinline!H!. Moreover, since the term \lstinline!M! is used in the
judgment \lstinline!aeq! within the context \lstinline!H!, we
associate \lstinline!M!  with an identity substitution (denoted by
\lstinline!..!).  In short, the directive allows us to lift the type
specified in ORBI to a contextual type which is meaningful in Beluga.
In fact, Beluga always needs additional information on how to
interpret terms---are they closed or can they depend on a given
context? For translating \verb+symG+ for example, we use the following
directive to indicate the dependence on the context:
\begin{verbatim}
% [bel] implicit M in H in symG.
% [bel] implicit N in H in symG.
\end{verbatim}

\subsection{Guidelines}

In addition, we introduce a set of \emph{guidelines} for ORBI specification
writers, with the goal of helping translators generate output that is
more likely to be accepted by a specific system.  ORBI 0.1 includes
four such guidelines, which are motivated by the desire not to put too
many constraints in the grammar rules.  First, as we have seen in our
examples, we use as a convention that free
variables which denote schematic variables in rules are written using upper
case identifiers; we use lower case identifiers for eigenvariables in
rules. Second, while the grammar does not restrict what types we can quantify
over, the intention is that we quantify over types of level-0, i.e.\ objects of
the syntax level, only. Third, in order to more easily accommodate
systems without dependent types, \verb!Pi! should not be used when
writing non-dependent types.  An arrow should be used instead.  (In
LF, for example, \verb!A -> B! is an abbreviation for \verb!Pi x:A.B!
for the case when \verb!x! does not occur in \verb!B!.  Following this
guideline means favoring this abbreviation whenever it applies.)
Fourth, when writing a context (grammar
\verb|ctx|), distinct variable names should be used in different blocks.


\section{Related Work}
\label{sec:related}

Our approach to structuring contexts of assumptions takes its
inspiration from Martin-L\"of's theory of judgments
\cite{MartinLof85}, especially in the way it has been realized in 
Edinburgh LF \cite{Harper93jacm}. However, our formulation owes more to Beluga's
type theory, where contexts are first-class citizens, than to the
notion of \emph{regular world} in Twelf. The latter was introduced in
\citet{Schurmann00phd}, and used in \citet{schurmann03tphols} for the
meta-theory of Twelf and in \citet{CSL} for different purposes. It was
further explicated in \citet{HarperLicata:JFP07}'s review of Twelf's
methodology, but its treatment remained unsatisfactory since the
notion of worlds is extra-logical. Recent work
\cite{Wang:2013} on a logical rendering of Twelf's totality checking
has so far been limited to closed objects.

The creation and sharing of a library of benchmarks has proven to be
very beneficial to the field it represents. The brightest example is
\emph{TPTP} \cite{TPTP}, whose influence on the development, testing and
evaluation of automated theorem provers cannot be
underestimated. Clearly our ambitions are much more limited. We have
also taken some inspiration from its higher-order extension \emph{THF0}
\cite{THF0}, in particular in its construction in stages. 

The success of {TPTP} has spurned other benchmark suites in related
subjects, see for example \emph{SATLIB} \cite{SATLIB}; however, the only one
concerned with induction is the \emph{Induction Challenge Problems}
(\url{http://www.cs.nott.ac.uk/~lad/research/challenges}), a
collection of examples geared to the \emph{automation} of inductive
proof. The benchmarks are taken from arithmetic, puzzles, functional
programming specifications etc.\ and as such have little connection
with our endeavor.
On the other hand both Twelf's wiki
(\url{http://twelf.org/wiki/Case_studies}), Abella's library
(\url{http://abella-prover.org/examples}) and Beluga's
distribution contain a set of context-intensive examples, some of
which coincide with the ones presented here. As such they are prime
candidates to be included in ORBI.

Other projects have put forward LF as a common ground:
\emph{Logosphere}'s goal (\url{http://www.logosphere.org}) was the
design of a representation language for logical formalisms, individual
theories, and proofs, with an interface to other theorem proving
systems that were somewhat connected, but the project never
materialized. \emph{SASyLF} \cite{SASyLF} originated as a tool to
teach programming language theory: the user specifies the syntax,
judgments, theorems \emph{and} proofs thereof (albeit limited to
\emph{closed} objects) in a paper-and-pencil HOAS-friendly way and the
system converts them to totality-checked Twelf code. The capability to
express and share proofs is of obvious interest to us, although such
proofs, being a literal proof verbalization of the corresponding Twelf type
family,
are irremediably verbose.

\emph{Why3} (\url{http://why3.lri.fr}) is a software verification platform
that intends to provide a front-end to third-party theorem provers,
from proof assistants such as Coq to SMT-solvers.  To this end Why3
provides a first-order logic with rank-1 polymorphism, recursive
definitions, algebraic data types and inductive predicates
\cite{why3}, whose specifications are then translated in the several
systems that Why3 supports. Typically, those translations are
forgetful, but sometimes, e.g., with respect to Coq, they add some
annotations, for example to ensure non-emptiness of types. Although we
are really not in the same business as Why3, there are several ideas
that are relevant, such as the notion of a \emph{driver}, that is, a
configuration file to drive transformations specific to a system. Moreover,
Why3 provides an API for users to write and implement their own drivers
and transformations.

\emph{Ott} \cite{ott} is a highly engineered tool for ``working
semanticists,'' allowing them to write programming language
definitions in a style very close to paper-and-pen specifications; 
then  those are compiled into \LaTeX\ and, more interestingly, into
proof assistant code, currently supporting Coq, Isabelle/HOL and
HOL\@. Ott's metalanguage is endowed with a rich theory of binders,
but at the moment it favors the ``concrete'' (non $\alpha$-quotiented)
representation, while providing support for the nameless
representation for a single binder. Conceptually, it would be natural
to extend Ott to generate ORBI code, as a bridge for Ott to support
HOAS-based systems. Conversely, an ORBI user would benefit from having
Ott as a front-end, since the latter notion of grammar and judgment
seems at first sight general enough to support the notion of schema
and context relation.

In the category of environments for programming language descriptions,
we mention \emph{PLT-Redex}
\cite{PLTbook} and also the \emph{K} framework \cite{RosuK}. In both, several
large-scale language descriptions have been specified and
tested. However, none of those systems has any support for binders, let alone
context specifications, nor can any meta-theory be carried out.

Finally, there is a whole research area dedicated to the handling and
sharing of mathematical content (\emph{MMK}
\url{http://www.mkm-ig.org}) and its representation (\emph{OMDoc}
\url{https://trac.omdoc.org/OMDoc}), which is only very loosely
connected to our project.

\section{Conclusion and Future Work}
\label{sec:concl}

We have presented
an initial set of benchmarks that highlight a variety of different
aspects of reasoning within a context of assumptions.  We have also
provided
an infrastructure for formalizing these benchmarks in a variety of
HOAS-based systems, 
and for facilitating their comparison.  We have developed a framework
for expressing contexts of assumptions as structured sequences, which
provides additional structure to contexts via schemas and
characterizes their basic  properties.  Finally, we have
designed (the initial version of) the ORBI (\underline{O}pen challenge
problem \underline{R}epository for systems supporting reasoning with
\underline{BI}nders) specification language, and created an open
repository of specifications, which initially contains the benchmarks
introduced in this paper.

Selecting a small set of benchmarks has an inherent element of
arbitrariness. The reader may 
 complain that there are many
other features and issues not covered in Sect.~\ref{sec:bench}. We
agree and we mention some additional categories, which  we could not
discuss in the present paper for the sake of space, but which will
(eventually) make it into the ORBI repository:
\begin{itemize}
\item One of the weak spots of most current HOAS-based
  systems is the lack of libraries, built-in data-types and related
  decision procedures: for example, case studies involving calculi of
  explicit substitutions require a small corpus of arithmetic facts,
  that, albeit trivial, still need to be (re)proven, while they could
  be automatically discharged by decision procedures such as Coq's
  \textsf{omega}.\footnote{Case in point, the strong normalization
    proof  for
    the $\lambda\sigma$ calculus in Abella, see 
    \url{http://abella-prover.org/examples/lambda-calculus/exsub-sn/},
    $15\%$ of which consists of basic facts about addition.}

\item There are also specifications that are \emph{functional} in
  nature, such as those that descend through the structure of a lambda
  term, say counting its depth, the number of bound occurrences of a
  given variable etc.; most HOAS systems would encode those functions
  relationally, but this entails again the additional proof obligations
  of proving those relations total and deterministic.

\item In the benchmarks that we have presented all blocks are composed
  of \emph{atoms}, but there are natural specifications, to wit the
  solution to the \textsc{PoplMark} challenge in \citet{Pientka07}, where
  contexts have more structure, as they are induced by
  \emph{third-order} specifications.
  For example, the rule for subtyping universally quantified types
  introduces  
  a non-atomic assumption  about transitivity, of the form:
  \begin{center}
    \lstinline!{a:tp}({U:tp} {V:tp} sub a U -> sub U V -> sub a V)!.
  \end{center}
 
\item Proofs by \emph{logical
    relations} 
  typically require, in order to define reducibility candidates,
  inductive definitions and strong function spaces, i.e., a function
  space that does not only model binding. A direct encodings of those
  proofs is out of reach for systems such as Twelf, although indirect
  encodings exist \cite{Schurmann:LICS08}. Other systems, such as
  Beluga and Abella, are well capable of encoding such proofs, but
  differ in how this is accomplished, see \citet{Cave:logrel14} and
  \citet{GacekMillerNadathur:JAR12}.

\item Finally, a subject that is gaining importance is the encoding of
  \emph{infinite} behavior, typically realized via some form of
  \emph{co-induction}. Context-intensive case studies have been
  explored for example in \citet{Howe}.
\end{itemize}

One of the outcomes of our framework for expressing contexts of
assumptions is the unified treatment of all
weakening/strengthening/exchange re-arrangements, via the
$\textsf{rm}$ and $\textsf{perm}$ operations. This opens the road to a
\emph{lattice-theoretic} view of declarations and contexts, where,
roughly, $x\preceq y$ holds iff $x$ can be reached from $y$ by some
$\mathsf{rm}$ operation: a generalized context will be the join of two
contexts and context relations can be identified by navigating the
lattice starting from the join of the to-be-related contexts. We plan
to develop this view and use it to convert G proofs into R and vice
versa, as a crucial step towards breaking the proof/type theory
barrier.

\smallskip

The description of ORBI given in Sect.~\ref{sec:mech} is best thought
of as a stepping stone towards a more comprehensive specification
language, much as \emph{THF0} \cite{THF0} has been extended to the
more expressive formalism $THF_i$, adding for instance, rank 1
polymorphism. Many are the features that we plan to provide in the
near future, starting from general (monotone) \emph{(co)inductive}
definitions; currently we only relate contexts, while it is clearly
desirable to relate arbitrary well-typed terms, as shown for example
in \citet{Cave:POPL12} and \citet{GacekMillerNadathur:JAR12} with
respect to normalization proofs. Further, it is only natural to
support infinite objects and behavior.  However, full support for
(co)induction is a complex matter, as it essentially
entails fully understanding the relationship between the proof-theory
behind Abella and Hybrid and the type theory of Beluga. Once this is
in place, we can ``rescue'' ORBI theorems from their current status as
comments and even include a notion of {proof} in
ORBI.

Clearly, there is a significant amount of implementation work ahead,
mainly on the ORBI2X tools side, but also on the practicalities
of the benchmark suite. Finally, we would like to open up the
repository to other styles of specification such nominal, locally
nameless etc.

\begin{acknowledgements}
The first and third author acknowledge the support of the Natural
Sciences and Engineering Research Council of Canada.
We  thank Kaustuv Chaudhuri, Andrew Gacek, Nada Habli, and
Dale Miller for discussing some aspects of this work with us.
The first author would also like to extend her gratitude to the
University of Ottawa's Women's Writers Retreats.
\end{acknowledgements}


\appendix
\section{Overview of Benchmarks}
\label{sec:overview}

In this appendix, we provide a quick reference guide to some of the
key elements of the benchmark problems discussed in
Section~\ref{sec:bench}.
In the tables below, ULC (STLC) stands for the untyped (simply-typed)
lambda-calculus, and POLY stands for the polymorphic lambda
calculus. The entry ``same'' means that there is no difference between the R and
G version of the theorem because there is only one context involved.

\subsection{A Recap of Benchmark Theorems}
    \begin{tabular}{llll}
    \textbf{Theorem} & \textbf{Thm No.}& \textbf{Version} & \textbf{Page}  \\
\hline

\textsf{aeq}-reflexivity for ULC   &\ref{thm:reflR} & R     &   \pageref{thm:reflR}     \\
\textsf{aeq}-reflexivity for ULC   &\ref{thm:reflG} & G     &
\pageref{thm:reflG}     \\
\textsf{aeq}-symmetry and transitivity for ULC   &\ref{le:sym-trans}
& same     &   \pageref{le:sym-trans}     \\
\textsf{atp}-reflexivity for POLY    &\ref{thm:atpreflG} & G     &
\pageref{thm:atpreflG}     \\
\textsf{aeq}-reflexivity for POLY    &\ref{thm:prelreflG} & G     &   \pageref{thm:prelreflG}     \\
\textsf{atp}-reflexivity for POLY    &\ref{thm:prelreflR_tp} & R     &
\pageref{thm:prelreflR_tp}     \\
\textsf{aeq}-reflexivity for POLY    &\ref{thm:prelreflR} & R     &   \pageref{thm:prelreflR}     \\
\textsf{aeq/deq}-completeness for ULC   &\ref{thm:ceqG} & G     &
\pageref{thm:ceqG}     \\    
\textsf{aeq/deq}-completeness for ULC   &\ref{thm:ceqR} & R     &
\pageref{thm:ceqR}     \\    
type uniqueness for STLC   &\ref{thm:unique} & same     &
\pageref{thm:unique}     \\ 
type preservation for parallel reduction for STLC    &\ref{thm:tps} & R     &
\pageref{thm:tps}     \\    
\textsf{aeq}-parallel substitution for ULC&\ref{thm:asubst} & same     &
\pageref{thm:asubst}   

\end{tabular}

\subsection{A Recap of Schemas and Their Usage}
    \begin{tabular}{llll}
    \textbf{Context} & \textbf{Schema} & \textbf{Block} & 
    \textbf{Description/Used in:} \\
\hline
    $\Phi_{\alpha}$    & $S_{\alpha}$     & $\istp{\alpha} $   &
    type variables     \\
    $ \Phi_{x}$      & $\isch$      & $\istm{x}$        & term variables
    \\
    $ \Phi_{\alpha x}$     & $\axsch$       & $\istp{\alpha} \alt
    \istm{x}$     & type/term variables     \\
 $    \atctx$      &  $\atpsch$    & $\istp{\alpha} \alt
 \istm{x};x\oftp T$     & type-checking for POLY     \\
        $\actx$  & $\asch $      & $\istm{x}; \termeq{x}{x}$   &
        Thm~\ref{thm:reflG},~\ref{le:sym-trans}, and~\ref{thm:asubst}      \\
    $\atpeqctx$      & $ \atpeqsch$     &
    $\istp{\alpha};\tpaeq{\alpha}{\alpha}$       & Thm~\ref{thm:atpreflG}    \\
    $\pctx$     & $\psch$      &
    $\istp{\alpha};\tpaeq{\alpha}{\alpha}~\alt~\istm{x};\termeq{x}{x}
    $   &  Thm~\ref{thm:prelreflG} \\
     $\Gamma_{da}$      &  $S_{da}$     & $ \istm x; \termequal x x;
     \termeq{x}{x}$   &       Thm~\ref{thm:ceqG}    \\
    $\Phi_{xd}$        & $S_{xd}$ & $\istm x; \termequal x x$  & Thm~\ref{thm:ceqR}\\
    $\Phi_{t}$       & $S_{t}$ & $\istm x; \oft x A$ &
    Thm~\ref{thm:unique}\\
    $\Phi_r$       &  $\rsch$     & $ \istm x; x \step x$  &     Thm~\ref{thm:tps}   
    \end{tabular}

\subsection{A Recap of the Main Context Relations and Their Usage}

    \begin{tabular}{lll}
    \textbf{Relation} &  \textbf{Related Blocks} & 
    \textbf{Used in:} \\
\hline
    $\crel \ictx \actx$        & $\crel {\istm x} {(\istm{x}; \termeq{x}{x})} $  & 
Thm~\ref{thm:reflR}    \\
$\crel{\alphctx}{\atpeqctx}$ & $\crel{\istp{\alpha}}
                        {(\istp{\alpha};\tpaeq{\alpha}{\alpha})}$ &
                        Thm~\ref{thm:prelreflR_tp}\\
$\crel{\axctx}{\pctx}$ &   $\crel \ictx \actx$ plus $\crel{\alphctx}{\atpeqctx}$ &
 Thm~\ref{thm:prelreflR}\\
 $\crel\actx\dctx$ & $\crel {(\istm{x};\termeq x x)}
                 {(\istm x; \termequal x x)}$ & 
Thm~\ref{thm:ceqR}\\
 $\crel\rctx\tctx$ & $\crel{ (\istm{x};x\step x)}
               {(\istm{x};x\oftp A)}$    & 
Thm~\ref{thm:tps}
    \end{tabular}

\section{ORBI Specification of Algorithmic and Declarative Equality}
\label{sec:aede-orbi}

The following ORBI specification provides a complete encoding of the
example of algorithmic vs.\ declarative equality used in
Subsections~\ref{ssec:basic},~\ref{ssec:sub}, and~\ref{ssec:exc}.

\begin{verbatim}
%% Syntax
tm: type.

app: tm -> tm -> tm.
lam: (tm -> tm) -> tm.

%% Judgments
aeq: tm -> tm -> type.
deq: tm -> tm -> type.

%% Rules
ae_a: aeq M1 N1 -> aeq M2 N2 -> aeq (app M1 M2) (app N1 N2).
ae_l: ({x:tm} aeq x x -> aeq (M x) (N x)) 
        -> aeq (lam (\x. M x)) (lam (\x. N x)).

de_a: deq M1 N1 -> deq M2 N2 -> deq (app M1 M2) (app N1 N2).
de_l: ({x:tm} deq x x -> deq (M x) (N x)) 
        -> deq (lam (\x. M x)) (lam (\x. N x)).
de_r: deq M M.
de_s: deq M1 M2 -> deq M2 M1.
de_t: deq M1 M2 -> deq M2 M3 -> deq M1 M3.

%% Schemas
schema xG: block (x:tm).
schema xaG: block (x:tm; u:aeq x x).
schema xdG: block (x:tm; u:deq x x).
schema daG: block (x:tm; u:deq x x; v:aeq x x).

%% Definitions
inductive xaR : {G:xG} {H:xaG} prop =
| xa_nil: xaR nil nil
| xa_cons: xaR G H -> xaR (G, block x:tm) (H, block x:tm; u:aeq x x).

inductive daR : {G:xdG} {H:xaG} prop =
| da_nil: daR nil nil
| da_cons: daR G H -> daR (G, block x:tm; v:deq x x)
                          (H, block x:tm; u:aeq x x).

%% Theorems
theorem reflG: {H:xaG} {M:tm} [H |- aeq M M].
theorem symG: {H:xaG}{M:tm}{N:tm} [H |- aeq M N] -> [H |- aeq N M].
theorem transG: {H:xaG}{M:tm}{N:tm}{L:tm}
              [H |- aeq M N] & [H |- aeq N L] -> [H |- aeq M L].
theorem ceqG: {G:daG} [G |- deq M N] -> [G |- aeq M N].
theorem substG: {H:xaG}{M1:tm->tm}{M2:tm}{N1:tm}{N2:tm}
  [H, block x:tm; aeq x x |- aeq (M1 x) (M2 x)] & [H |- aeq N1 N2] ->
  [H |- aeq (M1 N1) (M2 N2)].

theorem reflR : {G:xG}{H:xaG}{M:tm} xaR G H -> [H |- aeq M M].
theorem ceqR: {G:xdG}{H:xaG} daR G H -> [G |- deq M N] -> [H |- aeq M N].

%% Directives
% [hy,ab] wf tm.
% [hy,ab] explicit x in de_l.
% [hy,ab] explicit M in de_r.
% [hy,ab] explicit x in xG.
% [hy,ab] explicit x in xdG.
% [hy,ab] explicit x in daG.
% [hy,ab] explicit x in G in xaR.
% [hy,ab] explicit x in G in daR.
\end{verbatim}  

\end{document}